\PassOptionsToPackage{dvipsnames}{xcolor}
\documentclass[authorversion,acmsmall,nonacm,screen]{acmart}

\acmYear{2023}
\acmMonth{7}
\acmDOI{} %
\startPage{1}

\setcopyright{rightsretained}

\usepackage{booktabs}   %
\usepackage{subcaption} %
\usepackage[dvipsnames]{xcolor}
\usepackage{bcprules}
\usepackage{collcell}
\usepackage{listings}
\usepackage[scaled]{beramono}
\usepackage{lipsum}
\usepackage{mathpartir}

\usepackage{mathtools}
\usepackage{mdframed}
\usepackage{amsmath}
\usepackage{varioref}
\usepackage{cleveref}
\usepackage{scalerel}
\usepackage{xspace}
\usepackage{calc}
\usepackage{array,varwidth}
\usepackage{makecell}
\usepackage{diagbox}
\usepackage{float}
\usepackage{xfrac}
\usepackage{enumitem}
\usepackage[skip=2.5pt]{caption}
\usepackage{slashbox}
\usepackage{xifthen}
\usepackage{tikz}
\usetikzlibrary{arrows}

\makeatletter %
\def\arcr{\@arraycr}
\makeatother

\newcommand{\showDOI}[1]{\unskip}

\newtheorem{innercustomgeneric}{\customgenericname}
\providecommand{\customgenericname}{}
\newcommand{\newcustomtheorem}[2]{%
  \newenvironment{#1}[1]
  {%
   \renewcommand\customgenericname{#2}%
   \renewcommand\theinnercustomgeneric{##1}%
   \innercustomgeneric
  }
  {\endinnercustomgeneric}
}
\newcustomtheorem{re-definition}{Definition}
\newcustomtheorem{re-lemma}{Lemma}

\newcommand{\secref}[1]{Sec.~\ref{#1}} %

\newcommand{\Specsharp}{%
	{\settoheight{\dimen0}{C}Spec\kern-.05em \resizebox{!}{\dimen0}{\raisebox{\depth}{\#}}}}
\newcommand{\Csharp}{%
	{\settoheight{\dimen0}{C}C\kern-.05em \resizebox{!}{\dimen0}{\raisebox{\depth}{\#}}}}

\newcommand{\fun}[1]{\operatorname{#1}}
\newcommand{\DOM}{\fun{dom}}

\definecolor{blue-violet}{rgb}{0.54, 0.17, 0.89}
\definecolor{dark-cyan}{HTML}{135579}
\definecolor{magenta}{HTML}{a8264f}

\lstdefinelanguage{Neutral}%
{morekeywords={abstract,%
  case,catch,char,class,%
  def,else,extends,final,finally,for,%
  if,import,implicit,%
  match,module,%
  new,null,%
  object,override,%
  package,private,protected,public,%
  for,public,return,super,%
  this,trait,try,type,%
  val,var,%
  with,while,%
  yield,%
  let,end,%
	in,fun,alloc,inc%
  },%
  mathescape=true,%
  sensitive,%
  keywordstyle={\color{black}\bf\ttfamily},%
  commentstyle=\color{OliveGreen},%
  escapebegin=\color{OliveGreen},
  morecomment=[l]//,%
  morecomment=[s]{/*}{*/},%
  morecomment=[s][\color{darkgray}]{@}{\ },%
  morestring=[b]",%
  morestring=[b]',%
  showstringspaces=false%
}[keywords,comments,strings]%

\lstdefinelanguage{OOPSLA21}%
{morekeywords={abstract,%
  case,catch,char,class,%
  def,else,extends,final,finally,for,%
  if,import,implicit,%
  match,module,%
  new,null,%
  object,override,%
  package,private,protected,public,%
  for,public,return,super,%
  this,throw,trait,try,type,%
  val,var,%
  with,while,%
  yield,%
  let,end,%
	in,fun,alloc,inc%
  },%
  mathescape=true,%
  sensitive,%
  keywordstyle={\color{magenta}\bf\ttfamily},%
  commentstyle=\color{magenta},%
  escapebegin=\color{magenta},
  morecomment=[l]//,%
  morecomment=[s]{/*}{*/},%
  morecomment=[s][\color{magenta}]{@}{\ },%
  morestring=[b]",%
  morestring=[b]',%
  showstringspaces=false%
}[keywords,comments,strings]%

\lstdefinelanguage{PolyRT}%
{morekeywords={abstract,%
  case,catch,char,class,%
  def,else,extends,final,finally,for,%
  if,import,implicit,%
  match,module,%
  new,null,%
  object,override,%
  package,private,protected,public,%
  for,public,return,super,%
  this,throw,trait,try,type,%
  val,var,%
  with,while,%
  yield,%
  let,end,%
	in,fun,alloc,inc%
  },%
  mathescape=true,%
  sensitive,%
  keywordstyle={\color{dark-cyan}\bf\ttfamily},%
  commentstyle=\color{dark-cyan},%
  escapebegin=\color{dark-cyan},%
  morecomment=[l]//,%
  morecomment=[s]{/*}{*/},%
  morecomment=[s][\color{dark-cyan}]{@}{\ },%
  morestring=[b]",%
  morestring=[b]',%
  showstringspaces=false%
}[keywords,comments,strings]%

\newcommand{\untrack}[0]{^{\bot}}
\newcommand{\trackset}[1]{^{\texttt{\{#1\}}}}
\newcommand{\track}[0]{^{\varnothing}}
\newcommand{\trackvar}[1]{^{\texttt{#1}}}
\newcommand{\trackfresh}{^{\texttt{\ensuremath\vardiamondsuit}}}

\newcommand{\untrackl}[0]{^{\bot}}
\newcommand{\tracksetl}[1]{^{\tt{\{#1\}}}}

\newcommand{\trackl}[0]{^{\varnothing}}

\newcommand{\maybelang}{\ensuremath{\lambda^{\vardiamondsuit}}\xspace}
\newcommand{\polylang}{\ensuremath{\mathsf{F}_{<:}^{\vardiamondsuit}}\xspace}

\newcommand{\Type}[1]{\ensuremath{\mathsf{#1}}}
\newcommand{\Var}{\Type{Var}}

\newcommand{\TRef}{\Type{Ref}}
\newcommand{\TTop}{\Type{Top}}
\newcommand{\tref}{\mathsf{ref}}
\newcommand{\TUnit}{\Type{Unit}}
\newcommand{\tunit}{\mathsf{unit}}
\newcommand{\Loc}{\Type{Loc}}

\newcommand{\ty}[2][]{\ensuremath{\ifthenelse{\isempty{#1}}{#2}{#2^{\,#1}}}}

\newcommand{\TAll}[6]{\ensuremath{\forall f(\ty[#2]{#1} <: \ty[#4]{#3}) . \ty[#6]{#5}}}
\newcommand{\TLam}[5]{\ensuremath{\Lambda f(\ty[#2]{#1}) . {#5}}}
\newcommand{\TApp}[3]{\ensuremath{{#1}\ [\ty[#3]{#2}]}}

\newcommand{\ts}[1][]{\ensuremath{\ifthenelse{\isempty{#1}}{\,\vdash\,}{\,\vdash^{\,#1}\,}}}

\newcommand{\flt}{\ensuremath{\varphi}}

\newcommand{\cx}[2][]{\ensuremath{\ifthenelse{\isempty{#1}}{#2}{#2^{\,#1}}}}

\providecommand{\G}{G} %
\renewcommand{\G}[1][]{\cx[#1]{\Gamma}}

\newcommand{\HLBox}[2][teal!12]{\ensuremath{\mathchoice%
  {\setlength{\fboxsep}{.5ex}\colorbox{#1}{$\displaystyle#2$}}%
  {\setlength{\fboxsep}{.5ex}\colorbox{#1}{$\textstyle#2$}}%
  {\setlength{\fboxsep}{.5ex}\colorbox{#1}{$\scriptstyle#2$}}%
  {\setlength{\fboxsep}{.5ex}\colorbox{#1}{$\scriptscriptstyle#2$}}}}%

\newcommand{\QFresh}{\ensuremath{\vardiamondsuit}}
\newcommand{\qbot}{\ensuremath{\varnothing}}
\newcommand{\qfresh}{\ensuremath{\vardiamondsuit}}
\newcommand{\subq}{\ensuremath{\subseteq}}
\newcommand{\qlub}{\ensuremath{\cup}}
\newcommand{\qglb}{\ensuremath{\cap}}
\newsavebox{\SMALLSTAR}
\savebox{\SMALLSTAR}{\(\raisebox{.25ex}{\(\qfresh\)}\)}
\newcommand{\starred}[1]{\ensuremath{\mathord{\scalerel*{\usebox{\SMALLSTAR}}{*}}#1}}
\newsavebox{\OVRLP}
\savebox{\OVRLP}{$\raisebox{.37ex}[0pt][0pt]{$\mathrlap{\hspace{.415ex}\scaleobj{.5}{\vardiamondsuit}}$}\cap$}
\def\overlap{\ensuremath{\mathbin{\scalerel*{\usebox{\OVRLP}}{\sqcap}}}}

\newcommand{\qsat}[1]{\ensuremath{#1\mathord{*}}}
\newcommand{\WF}[1]{\ensuremath{#1\ \mathsf{ok}}}
\newcommand{\reaches}{\ensuremath{\mathrel{\leadsto}}}

\newcommand{\BOX}[1]{\fbox{$\strut #1$}}

\newcommand{\FV}{\ensuremath{\operatorname{fv}}}
\newcommand{\FTV}{\ensuremath{\operatorname{ftv}}}

\newcommand{\vgap}{\vspace{7pt}}

\newcommand{\bao}[1]{{\color{magenta}#1}}
\newcommand{\baotype}[2]{{\color{magenta}{\lstinline[language=OOPSLA21,mathescape=true]{#1}\ensuremath{{#2}}}}}
\newcommand{\baoterm}[1]{{\lstinline[language=OOPSLA21]{#1}}}

\newcommand{\hole}[1]{\ensuremath{[\,#1\,]}}
\newcommand{\CX}[3][black]{\ensuremath{{\color{#1}#2\ifthenelse{\isempty{#3}}{}{\hole{{\color{black}#3}}}}}}

\newcommand{\ie}{{\em i.e.}\xspace}

\newcommand{\rulename}[1]{(\textsc{#1})} %

\newcommand{\Fsub}{\ensuremath{\mathsf{F}_{<:}}\xspace}

\newcommand{\artifact}{\url{https://github.com/TiarkRompf/reachability}\xspace}

\lstdefinelanguage{DOT}%
{morekeywords={val,new},%
  sensitive,%
  morecomment=[l]//,%
  morecomment=[s]{/*}{*/},%
  morestring=[b]",%
  morestring=[b]',%
  showstringspaces=false%
}[keywords,comments,strings]%

\newlength{\trulemargin}
\newlength{\trulewidth}
\newlength{\srulewidth}
\newenvironment{trules}{$\vspace{0.5em}\ba{p{\trulemargin}@{~}p{\trulewidth}@{~}p{\trulemargin}}}{\ea$}
\newenvironment{srules}{$\vspace{0.5em}\ba{p{\trulemargin}@{~}p{\srulewidth}}}{\ea$}

\newcommand{\ba}{\begin{array}}
\newcommand{\ea}{\end{array}}

\newcommand{\ei}{\end{array}}
\newcommand{\bcases}{\left\{\begin{array}{ll}}
\newcommand{\ecases}{\end{array}\right.}

\newcommand{\eg}{{\em e.g.}\xspace}

\newcommand{\dom}{\mbox{\sl dom}}

\newcommand{\judgement}[2]{{\textsf{\textbf{#1}}} \hfill #2}

\begin{document}

\title[Polymorphic Reachability Types]{Polymorphic Reachability Types: Tracking Freshness, Aliasing, and Separation in Higher-Order Generic Programs}

\author{Guannan Wei}
\affiliation{
  \department{Department of Computer Science}              %
  \institution{Purdue University}            %
  \city{West Lafayette}
  \state{IN}
  \country{USA}                    %
}
\email{guannanwei@purdue.edu}          %

\author{Oliver Bra\v{c}evac}
\affiliation{
  \department{Department of Computer Science}              %
  \institution{Purdue University}            %
  \city{West Lafayette}
  \state{IN}
  \country{USA}                    %
}
\email{bracevac@purdue.edu}          %

\author{Songlin Jia}
\affiliation{
  \department{Department of Computer Science}              %
  \institution{Purdue University}            %
  \city{West Lafayette}
  \state{IN}
  \country{USA}                    %
}
\email{jia137@purdue.edu}          %

\author{Yuyan Bao}
\affiliation{
  \department{School of Computer and Cyber Sciences}              %
  \institution{Augusta University}            %
  \city{Augusta}
  \state{GA}
  \country{USA}                    %
}
\email{yubao@augusta.edu}          %

\author{Tiark Rompf}
\affiliation{
  \department{Department of Computer Science}              %
  \institution{Purdue University}            %
  \city{West Lafayette}
  \state{IN}
  \country{USA}                    %
}
\email{tiark@purdue.edu}          %

\begin{abstract}
Fueled by the success of Rust, many programming languages are
adding substructural features to their type systems. The promise of tracking
properties such as lifetimes and sharing is tremendous, not just for
low-level memory management, but also for controlling higher-level resources
and capabilities. But so are the difficulties in adapting successful
techniques from Rust to higher-level languages,
where they need to interact with other advanced features,
especially various flavors of functional and type-level abstraction.
Hence, recent proposals such as Scala's Capture Types target far
narrower domains than Rust. But what would it take
to bring full-fidelity reasoning about lifetimes and sharing to
mainstream languages? Reachability types are a recent proposal
that has shown promise in scaling to higher-order but
monomorphic settings, tracking aliasing and separation on top of a
substrate inspired by separation logic.
The \(\lambda^*\) reachability type system qualifies types with sets of
reachable variables and guarantees separation if two terms have disjoint
qualifiers.
However, naive extensions with type polymorphism and/or precise reachability
polymorphism are unsound, making \(\lambda^*\) unsuitable for adoption in real languages.
Combining reachability and type polymorphism that is precise, sound, and
parametric remains an open challenge.

This paper presents a rethinking of the design of reachability tracking and
proposes a solution to the key challenge of reachability polymorphism.
Instead of always tracking the transitive closure of reachable variables
as in the original design, we only track variables reachable in a single step
and compute transitive closures only when necessary, thus preserving
chains of reachability over known variables that can be
refined using substitution.
To enable this property, we introduce a new freshness qualifier, which
indicates variables whose reachability sets may grow during evaluation
steps.
These ideas yield the simply-typed \maybelang-calculus with precise
lightweight, \ie, quantifier-free, reachability polymorphism, and
the \polylang-calculus with bounded parametric polymorphism over types
and reachability qualifiers.
We prove type soundness and a preservation of separation property in Coq.
We show that our system subsumes both previous reachability type
systems as well as the essence of Scala's capture types, making true
tracking of lifetimes and sharing practical for mainstream languages.

 \end{abstract}


\maketitle

\lstMakeShortInline[keywordstyle=,%
                    flexiblecolumns=false,%
                    language=PolyRT,
                    basewidth={0.56em, 0.52em},%
                    mathescape=true,%
                    basicstyle=\footnotesize\ttfamily]@

\section{Introduction} \label{sec:intro}

Type systems based on ownership and borrowing are seeing increasing
practical adoption, most prominently for ensuring memory safety in
comparatively low-level ``systems languages'' such as Rust \cite{DBLP:conf/sigada/MatsakisK14}.
But what about higher-level languages, specifically
those that rely to a larger degree on functional and type-level
abstraction (\eg, Scala and OCaml)?

Tracking substructural properties such as lifetimes and sharing
in the type system holds great promise, not only for low-level
memory management, but also for managing a variety of other
resources (\eg, files, network sockets, access tokens, mutex locks, etc.), for
tracking effects (\eg, via capabilities for exceptions, algebraic
effects, continuations, callbacks via async/await, etc.), as
well as for compiler optimizations (\eg, fine-grained dependency
analysis~\cite{oopsla23}, safe destructive updates, etc.).
Therefore, it is no surprise that several mainstream languages
are moving in this direction with experimental proposals backed
by serious engineering efforts, which are to a large degree inspired by
the success of Rust (\eg, Linear Haskell
\cite{DBLP:journals/pacmpl/BernardyBNJS18} and Scala Capture Types
\cite{DBLP:journals/corr/abs-2105-11896, DBLP:conf/scala/OderskyBBLL21, DBLP:journals/corr/abs-2207-03402}).

However, these proposals all focus on relatively narrow substructural
properties rather than attempting to model lifetimes
and sharing with similar generality as Rust's ownership and
borrowing approach.
For example, the Linear Haskell extension specifically tracks multiplicity of
uses, and the Scala Capture Types extension specifically targets
effect capabilities.
Of course, this is neither neglect, nor coincidence, but the observable effect
of an underlying hard problem: ownership type systems
\cite{noble1998flexible,DBLP:conf/oopsla/ClarkePN98, DBLP:series/lncs/ClarkeOSW13}
that would enable tracking more sophisticated lifetime properties traditionally
rely on strict heap invariants (selectively relaxed via borrowing
\cite{hogg1991islands}) that are difficult to enforce in the presence of
pervasive functional and type-level abstraction (see \Cref{fig:counter}
for an example).

\subsubsection*{Reachability Types}

Reachability types~\cite{DBLP:journals/pacmpl/BaoWBJHR21} are a recently
introduced close cousin to ownership types and nephew to separation logic
\cite{DBLP:conf/lics/Reynolds02, DBLP:conf/csl/OHearnRY01},
which have shown potential to bring more of the benefits of ownership type
systems to high-level languages.
The key idea of reachability types is to track reachability and aliases
as type qualifiers, which is best demonstrated by a code example with
ML-style references (types shown as comments):
\begin{lstlisting}[language=OOPSLA21]
  val x = new Ref(0)  // : Ref[Int]$\tracksetl{x}$
  val y = x           // : Ref[Int]$\tracksetl{x, y}$
\end{lstlisting}
Qualifiers are sets of identifiers attached to types.
Variable @x@ is bound to a \emph{freshly} allocated reference; its type
qualifier tracks only @x@ itself.  When @y@ is bound to @x@, the type
qualifier of @y@ tracks both @x@ and @y@, indicating that both can be
reached from @y@.

Previous work~\cite{DBLP:journals/pacmpl/BaoWBJHR21} has shown how
reachability types elegantly support functional abstraction beyond what
is available in Rust. For example,
\Cref{fig:counter} shows a program with escaping functions that can track the
sharing of locally-defined resources, which cannot be expressed under
Rust's ``shared XOR mutable'' constraint.
In \Cref{fig:counter}, we define a @counter@ function that returns a
pair of functions to increase or decrease a mutable variable.
Both the ``increase'' and ``decrease'' functions \emph{capture} the local heap-allocated
reference cell @c@ and \emph{escape} from @c@'s defining scope.
Once escaped, the name @c@ is not meaningful in the outer scope.
Reachability types use the outer function's self-reference @p@ to model this
escaping behavior and preserve the tracking of shared resources.
In contrast, Rust does not allow two functions to capture the same variable in
a mutable way, unless using dynamic reference counting to bypass the static
ownership discipline.

\begin{figure}[t]
\begin{lstlisting}[language=Neutral]
def counter(n: Int) = {             // counter: Int => $\mu$p.Pair[(()=>Unit)$\tracksetl{p}$, (()=>Unit)$\tracksetl{p}$]$\trackl$
  val c = new Ref(n)                //  : Ref[Int]$\tracksetl{c}$
  (() => c += 1, () => c -= 1)      //  : Pair[(()=>Unit)$\tracksetl{c}$, (()=>Unit)$\tracksetl{c}$]$\tracksetl{c}$
}
                                    // instantiate the self-reference $\mathit{p}$ with bound name $\mathit{ctr}$:
val ctr = counter(0)                //  : Pair[(()=>Unit)$\tracksetl{ctr}$, (()=>Unit)$\tracksetl{ctr}$]$\tracksetl{ctr}$
                                    // name $\mathit{ctr}$ abstracts over its captured variables:
val incr = fst(ctr)                 //  : (()=>Unit)$\tracksetl{ctr}$
val decr = snd(ctr)                 //  : (()=>Unit)$\tracksetl{ctr}$
\end{lstlisting}
\caption{
  An example (from \cite{DBLP:journals/pacmpl/BaoWBJHR21})
  demonstrating first-class functions supported by reachability types.
  The \lstinline{counter} function returns two closures over a shared
  mutable reference (which is a fresh value before binding it to \lstinline{c}).
  The return value is a pair typed with a self-reference \lstinline{p}
  to express the capture of \lstinline{c} by both closures.
  The self-reference introduced by the $\mu$-notation is similar to DOT,
  but reachability types desugar it into function types (cf.  \Cref{sec2:polydata}
  for the encoding and \Cref{sec:maybe:syntax} for the formal syntax).
  Rust's type system prevents returning closures over local mutable references due to the
  ``shared XOR mutable'' restriction, and has to resort to dynamic reference counting
  to implement similar functionality.
}
\label{fig:counter}
\end{figure}

The idea of tracking reachability at the type level gives rise to powerful
reasoning capabilities --- most importantly, when considering the absence of
reachability, namely \emph{separation}. Two terms are separate when their type
qualifiers are disjoint. The metatheory of reachability types guarantees not
only preservation of types but also preservation of separation: if two
expressions have disjoint qualifiers, they will evaluate to disconnected object
graphs at runtime.
Taking reachability and separation as the fundamental building blocks
of a type system stands in contrast to traditional ownership type systems
that put heap invariants about unique access paths first and selectively
re-introduce sharing via borrowing.
Crucially, reachability and separation appear as more fundamental properties in
the sense that formal accounts of Rust's type
system~\cite{DBLP:journals/pacmpl/0002JKD18} are typically expressed using
separation logic as the meta-language.

\subsubsection*{Limitations of $\lambda^*$} While the reachability type system $\lambda^*$
presented by \citeauthor{DBLP:journals/pacmpl/BaoWBJHR21} has shown key
advances with regards to reasoning about lifetimes and sharing in the presence
of functional abstraction, there are still significant gaps on the way to
smoothly integrating reachability types into real-world high-level languages
such as Scala and OCaml, especially regarding type abstraction
and polymorphic data types that are missing from $\lambda^*$.

The key obstacle on the way is $\lambda^*$'s treatment of untracked values and fresh values.
In real-world programs, not all values need to have their reachability tracked (\eg,
pure functions and non-resource values).
However, keeping tracked and untracked values apart is difficult when crossing
abstraction boundaries.
While the $\lambda^*$-calculus does support
untracked values, it conflates untracked values and fresh values, where fresh
values are tracked values (\eg, allocations) but have not been bound to known
variables.
In $\lambda^*$, untracked values can be upcast to be fresh. Although being
sound, it comes at a loss in precision, which means that code cannot be generic
over the tracking status of arguments (see \Cref{sec:revisit} for detailed
examples).
This conflation of untracked and fresh values in the $\lambda^*$-calculus
unfortunately leads to the following important limitations in expressiveness:

\begin{itemize}[leftmargin=2em]
  \item The type system provides no abstraction over tracked and untracked types,
    \eg, a function cannot work over both tracked and untracked arguments
    while faithfully tracking their reachability.
  \item Qualifier-dependent function applications can only have shallow
    dependencies, \ie, the argument name can only occur in the outermost
    qualifier of the return type.
  \item Partly due to the first restriction, $\lambda^*$ does not support parametric
    polymorphism for either types or qualifiers.
  \item Due to the conflation of untracked and fresh values, $\lambda^*$ cannot
    support nested mutable references in the base system without extending it
    first with a flow-sensitive effect system and move semantics.
\end{itemize}
This paper overcomes the above limitations of reachability types and proposes
new variants of reachability types that track fine-grained lifetime properties
for higher-order, imperative, and polymorphic languages.
By rethinking reachability tracking and proposing a novel notion of
freshness, we show how systems like $\lambda^*$ can smoothly support
precise reachability polymorphism and type abstraction.

\begin{figure}[t]
  \centering
  \begin{subfigure}[b]{0.32\textwidth}
    \centering
    \includegraphics[trim={1cm 1cm 2.5cm 1cm},clip,width=0.9\textwidth]{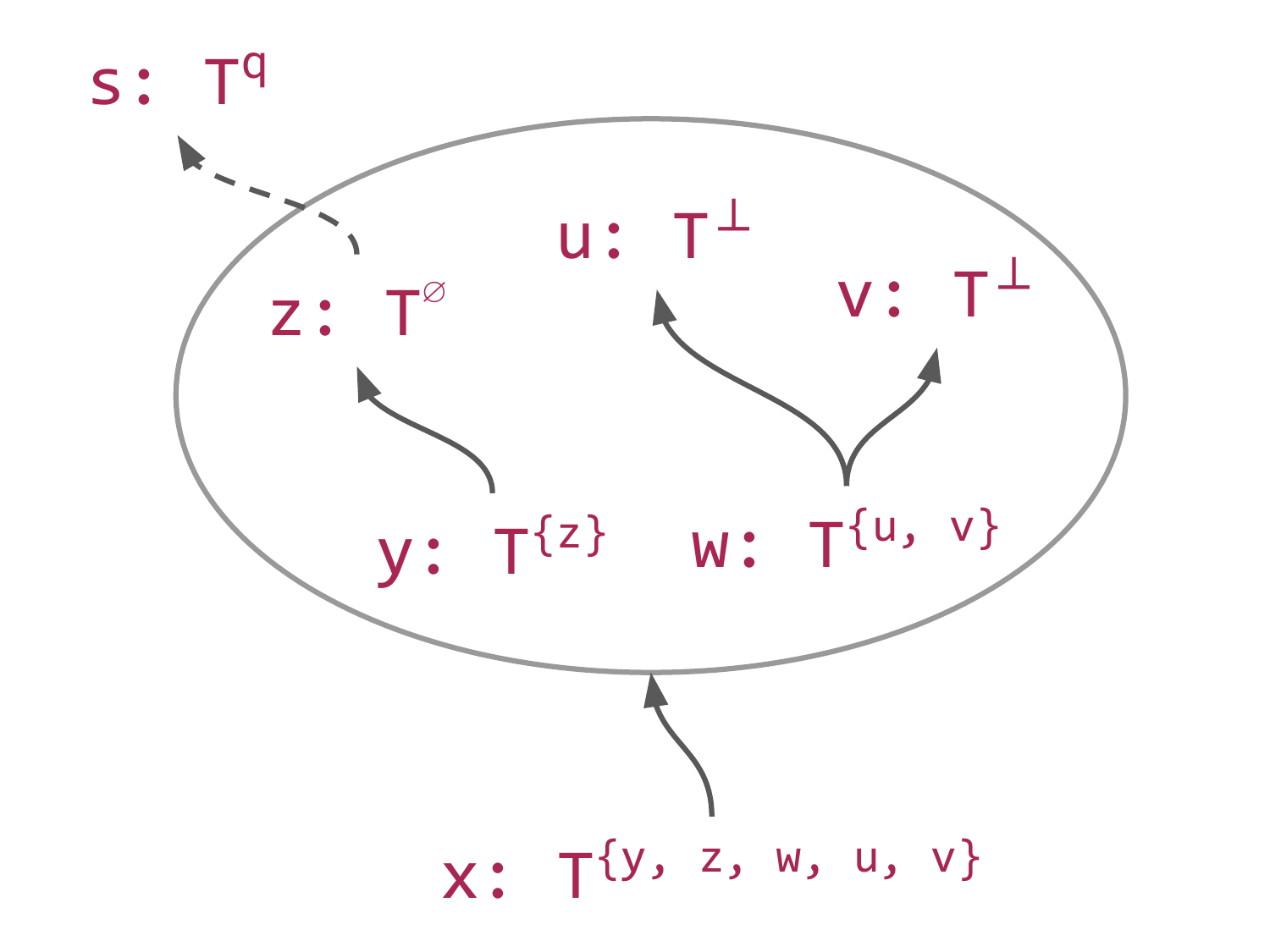}
    \caption{$\lambda^*$ \cite{DBLP:journals/pacmpl/BaoWBJHR21} tracks all reachable variables transitively.
    Leaf nodes are untracked ($\bot$ in $\lambda^*$). }
    \label{fig:transreach}
  \end{subfigure}
  \hfill
  \begin{subfigure}[b]{0.32\textwidth}
    \centering
    \includegraphics[trim={1cm 1cm 2cm 1cm},clip,width=0.9\textwidth]{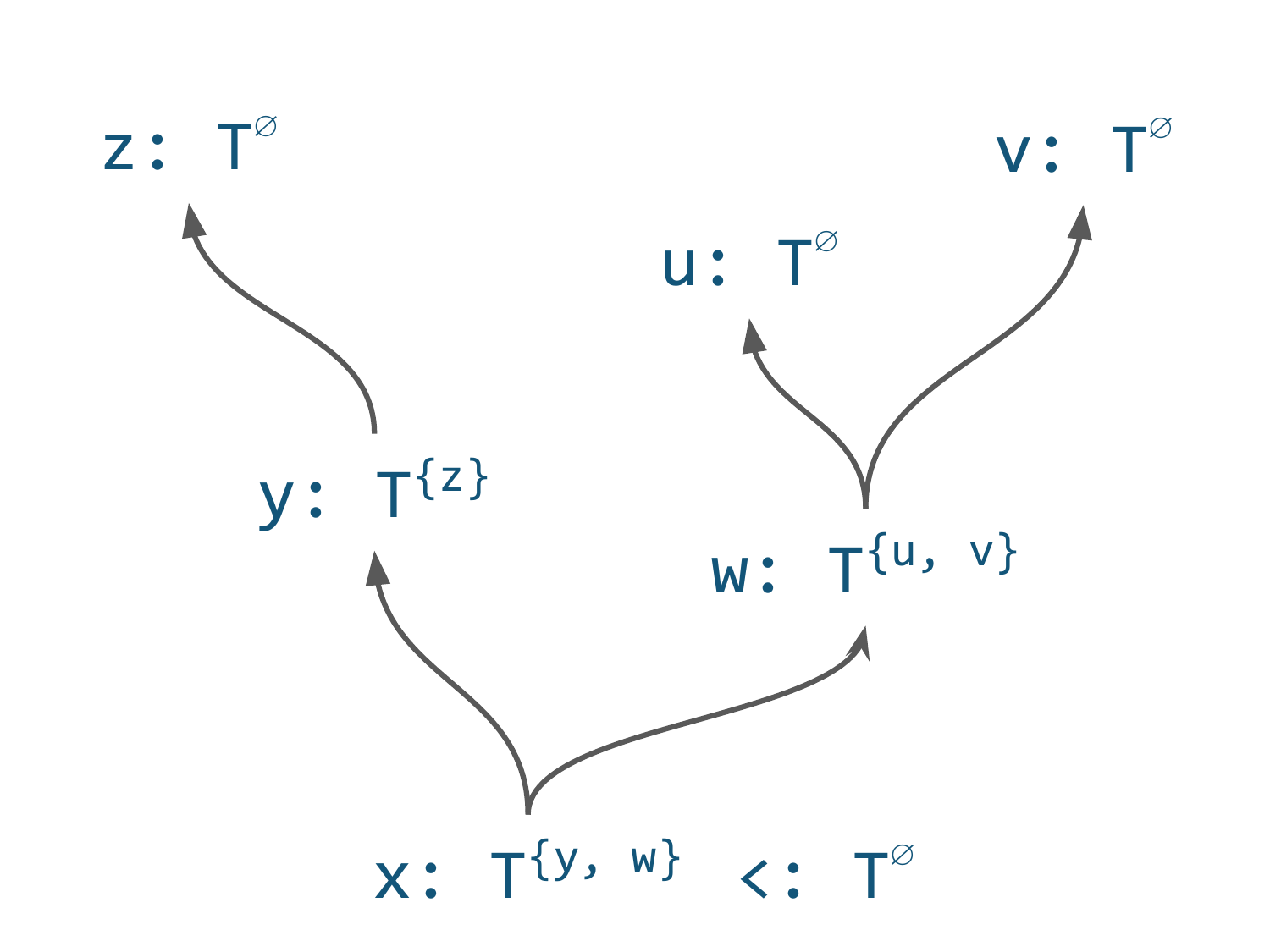}
    \caption{$\maybelang$ tracks one-step reachability by default.
    \texttt{x: T}$\trackset{y, w}$ can be upcast to untracked \texttt{x: T}$\track$.}
    \label{fig:onestepreach}
  \end{subfigure}
  \hfill
  \begin{subfigure}[b]{0.32\textwidth}
    \centering
    \includegraphics[trim={1cm 1cm 2cm 1cm},clip,width=0.9\textwidth]{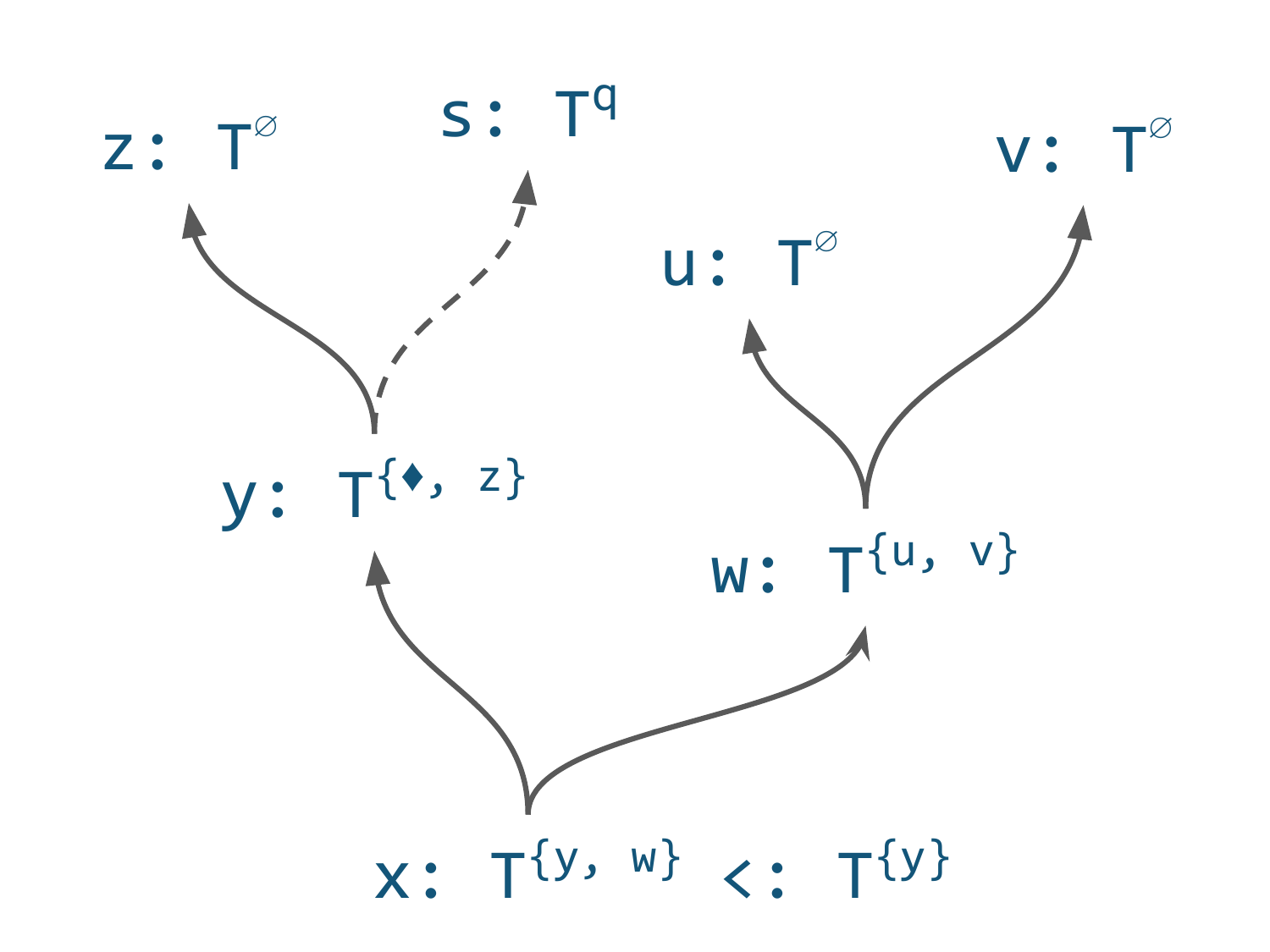}
    \caption{$\maybelang$ models freshness using the $\vardiamondsuit$ marker, which
    prevents further upcasting via subtyping (beyond \texttt{y}).}
    \label{fig:onestepfresh}
  \end{subfigure}
  \vspace{1ex}
  \caption{Illustration and comparison of different reachability tracking
  mechanisms.
  We use solid lines for direct reachability, and dashed lines for
  reachability that is unobservable in the current context (cf. \Cref{sec:freshness}).
  \ref{fig:transreach} illustrates prior work by \citet{DBLP:journals/pacmpl/BaoWBJHR21},
  \ref{fig:onestepreach} reflects both this work and Scala capture types
  \cite{DBLP:journals/corr/abs-2207-03402},
  and \ref{fig:onestepfresh} illustrates the unique feature of this work,
  which prevents upcasting through ``fresh'' variables, and thus allows substituting fresh variables with \emph{larger} (but observably separate)
  reachability sets during evaluation.
  Thus, this work subsumes the essential aspects of both $\lambda^*$ (separation) and capture types (qualifier refinement using subtyping).
  }
  \vspace{-3ex}
  \label{fig:one_step_vs_trans}
\end{figure}

\subsubsection*{Preserving Chains of Reachability}
In \citeauthor{DBLP:journals/pacmpl/BaoWBJHR21}'s system, qualifiers are
assigned to include all transitively reachable variables, \ie, the reachability
sets are eagerly saturated.  However, this is not always necessary and leads to
precision loss when the reachability set of a variable is refined to a smaller
set using substitution. If the reachability set containing the variable is
transitively saturated, the now superfluous elements cannot be removed, unless
one would recompute the transitive closure from scratch.

In contrast, the new design proposed in this paper tracks one-step reachability
by default, and only computes transitively saturated reachability sets on
demand, \eg, before computing intersections to check separation (see
\Cref{sec2:saturation}).  The two different mechanisms are illustrated in
\Cref{fig:one_step_vs_trans}.  In \Cref{fig:onestepreach},
{\color{dark-cyan}@x@} only tracks its immediate reachable variables, namely
{\color{dark-cyan}@{y, w}@}, whereas in \Cref{fig:transreach}, $\lambda^*$
tracks all variables that can be transitively reached from
{\color{magenta}@x@}.

Importantly, tracking one-step reachability preserves chains of reachability,
which allow us to maintain higher precision across substitution, both as part
of dependent function application and during reduction steps.  This approach
yields a new notion of ``maybe-tracked'' values, whose tracking status solely
depends on other variables from the context.  For example, by refining
reachability through the subtyping relation (see \Cref{sec:motiv:maybe}), @x@'s
reachability set in \Cref{fig:onestepreach} can be ``upcast'' to the empty set,
which precisely reflects its true untracked status.

\subsubsection*{A New Freshness Notion}
In \citeauthor{DBLP:journals/pacmpl/BaoWBJHR21}'s system, untracked values are
represented with the $\bot$ qualifier and fresh values with the $\varnothing$
qualifier, indicating an empty set of reachable variables.  Fresh values are
tracked but not observably aliased in the context.  Untracked values can be
upcast to fresh values, but not vice versa. Treating a tracked value as
untracked would be a soundness violation.

Instead of classifying values as \emph{untracked} or \emph{tracked}, we propose
to classify them as \emph{potentially fresh} or \emph{definitely non-fresh}.
To this end, we introduce an explicit freshness marker $\vardiamondsuit$ in
qualifiers for fresh values. With that addition, untracked values are naturally
assigned the empty reachability set, thus eliminating $\bot$ in the new system.
The freshness marker in qualifiers indicates that the expression may reach
unobservable variables or locations, which will materialize during evaluation.
Since $\vardiamondsuit$ signifies a statically unknown reachability set, it
serves as a barrier in subtyping chains, so that one cannot upcast beyond
$\vardiamondsuit$.
\Cref{fig:onestepfresh} shows such an example where upcasting is blocked
by the freshness marker on @y@. However, it is still possible to eliminate @w@ in
the qualifier @{y,w}@ since its leaf nodes are in fact untracked entities.

\subsubsection*{Polymorphic Reachability Types}
With the new mechanism for tracking reachability chains and the new freshness notion, we
present the $\maybelang$-calculus based on the simply-typed $\lambda$-calculus.
Compared to \citeauthor{DBLP:journals/pacmpl/BaoWBJHR21}'s $\lambda^*$, the
$\maybelang$-calculus addresses the fundamental expressiveness
limitations in $\lambda^*$ from above:
$\maybelang$ features precise reachability polymorphism without explicit
quantification, it supports deep dependencies in qualifier-dependent
applications, and supports nested references.
Furthermore, on top of the $\maybelang$-calculus, we develop extensions
with bounded quantification over types and qualifiers, leading to the
$\polylang$-calculus that can express polymorphic data types.
Polymorphic data such as pairs in $\polylang$ track precise
reachability of their components, which is not supported in $\lambda^*$.

\subsubsection*{Contributions}
We summarize our contributions as follows:
\begin{itemize}[leftmargin=2em]
  \item We identify the root issues in prior work leading to imprecise
    reachability tracking and address them by preserving transitive chains of
    reachability based on a more explicit ``freshness'' representation.
    We explain the key ideas and demonstrate the new type system informally
    with examples (\Cref{sec:motiv}).
  \item We present the formal theory and metatheory of (1) the $\maybelang$-calculus with precise
    reachability polymorphism that improves over
    \citet{DBLP:journals/pacmpl/BaoWBJHR21}'s $\lambda^*$-calculus
    (\Cref{sec:maybe}), and (2) the $\polylang$-calculus with bounded
    type-and-qualifier abstraction as an \Fsub-style extension of $\maybelang$
    (\Cref{sec:poly}). We prove type soundness and preservation of separation
    property for both calculi.
  \item We demonstrate that our system enables richer expressiveness in
    programming with capabilities compared to Scala capture types (\Cref{sec:scalacapture}).
    Our system thus subsumes both the original reachability types $\lambda^*$ \cite{DBLP:journals/pacmpl/BaoWBJHR21}
    and the essence of capture types \cite{DBLP:journals/corr/abs-2207-03402}.
  \item We have mechanized the metatheory of \maybelang and \polylang in Coq, including
  all the results and examples in this paper.
      We have also implemented a prototype implementation that can
    typecheck the examples in the paper.
    The Coq mechanization and prototype can be found at \artifact.
\end{itemize}
\Cref{sec:related} discusses related work and \Cref{sec:conclusion} concludes the paper.

\section{Key Ideas and Motivating Examples} \label{sec:motiv}

\newcommand{\tabcode}[1]{\lstinline[language=OOPSLA21,basicstyle=\tt\footnotesize]!#1!}
\newcommand{\tabcodepoly}[1]{\lstinline[language=PolyRT,basicstyle=\tt\footnotesize]!#1!}

\begin{table}[t]\footnotesize
\caption{Overview and comparison of $\lambda^*$ and this work. ``--'' indicates
  there is no equivalent notion in the system.
  The \lstinline{id} function is the polymorphic identity function as
  defined in the respective system.
  }\label{tab:compare}
\begin{tabular}{lll}
\toprule
                           & $\lambda^*$ \cite{DBLP:journals/pacmpl/BaoWBJHR21}   & This work \\[.5ex]
\midrule
\begin{tabular}[c]{@{}l@{}}
  \textbf{Untracked}\rule{0pt}{2.6ex}\\
  Primitive/atomic values
\end{tabular}
                           & \begin{tabular}[c]{@{}l@{}}
                             $\ty[\bot]{T}$\rule{0pt}{2.6ex} \\
                             \tabcode{val x = 42 //$\;$: Int$\untrack$}
                             \end{tabular}
                           & \begin{tabular}[c]{@{}l@{}}
                             $\ty[\qbot]{T}$\rule{0pt}{2.6ex}\\
                             \tabcodepoly{val x = 42 //$\;$: Int$\track$}
                             \end{tabular} \\[2.5ex] %
\begin{tabular}[c]{@{}l@{}}
  \textbf{Reachability Assignment} \\
  Transitive closure vs. \\
  immediate reachability \\
\end{tabular}
                           & \begin{tabular}[c]{@{}l@{}}
                             Reflexive \& transitive\\
                             \tabcode{val z = x //$\;$z : T$\trackset{z,x,...}$}\\
                             \\
                             \end{tabular}
                           & \begin{tabular}[c]{@{}l@{}}
                             One-step by default, transitive \\
                             on demand (\secref{sec2:ondemand}) \\
                             \tabcodepoly{val z = x //$\;$z : T$\trackset{z}$}
                           \end{tabular}  \\[3.5ex] %
\begin{tabular}[c]{@{}l@{}}
  \textbf{Fresh and Tracked}\\
  {Tracked but unbound in the context}
\end{tabular}
                           & \begin{tabular}[c]{@{}l@{}}
                             $\ty[\qbot]{T}$\\
                             \tabcode{alloc() : T$\track$}
                             \end{tabular}
                           & \begin{tabular}[c]{@{}l@{}}
                             $\ty[\{\QFresh, ...\}]{T}$ (\secref{sec:freshness}) \\
                             \tabcodepoly{alloc() : T$\trackfresh$}
                             \end{tabular} \\[2.5ex] %
\begin{tabular}[c]{@{}l@{}}
  \textbf{Reachability Polymorphism} \\
  {Functions preserving reachability} \\
  {that depends on arguments} \\
  $~$
\end{tabular}
                           & \begin{tabular}[c]{@{}l@{}}
                             Non-parametric \& imprecise \\
                             (\secref{sec:motiv:limit}) \\
                             \tabcode{id(42) : Int$\track$} \\
                             \tabcode{id(alloc()) : Int$^\varnothing$}
                             \end{tabular}
                           & \begin{tabular}[c]{@{}l@{}}
                             Parametric \& precise \\
                             (\secref{sec:motiv:precise}) \\
                             \tabcodepoly{id(42) : Int$\track$} \\
                             \tabcodepoly{id(alloc()) : Int$^\QFresh$}
                             \end{tabular} \\[5ex] %
\begin{tabular}[c]{@{}l@{}}
  \textbf{Qualifier Subtyping}\\
  {How qualifiers can be upcast} \\
  $~$ \\
  $~$
\end{tabular}
                           & \begin{tabular}[c]{@{}l@{}}
                             Set inclusion \\
                             $T^{q_1} <: T^{q_2}$ if $q_1 \subseteq q_2$ \\
                             $~$\\
                             $~$
                             \end{tabular}
                           & \begin{tabular}[c]{@{}l@{}}
                             Context dependent (\secref{sec:motiv:maybe})\\
                             $\Gamma =$ \tabcodepoly{x: T$\track$, y: T$\trackfresh$} \\
                             $\Gamma \vdash$ \tabcodepoly{T$\trackset{x} <:\hspace{0.5em}$ T$\track$} \\
                             $\Gamma \vdash$ \tabcodepoly{T$\trackset{y} \not{<:}\hspace{0.5em}$ T$\trackfresh$}
                             \end{tabular} \\[5ex] %
\begin{tabular}[c]{@{}l@{}}
  \textbf{``Maybe'' Tracked}\\
  {Variable-dependent tracking status} \\
\end{tabular}
                           &  \begin{tabular}[c]{@{}l@{}}-- \\ \ \end{tabular}
                           &  \begin{tabular}[c]{@{}l@{}}
                              $\ty[q]{T}$ if $\QFresh \notin q$ (\secref{sec:motiv:maybe}) \\
                              $\Gamma \vdash$ \tabcodepoly{T$\trackset{x} \equiv\hspace{0.5em}$ T$\track$} \\
                              \end{tabular} \\[2.5ex] %
\begin{tabular}[c]{@{}l@{}}
  \textbf{Transitive Reachability} \\
  {When transitive closure is used}
\end{tabular}
                           &  \begin{tabular}[c]{@{}l@{}}
                              Always saturated \\
                              $~$
                              \end{tabular}
                           &  \begin{tabular}[c]{@{}l@{}}
                              On-demand when \\
                              checking overlap (\secref{sec2:saturation})
                              \end{tabular} \\[2.5ex] %
\begin{tabular}[c]{@{}l@{}}
  \textbf{Qualifier-Dependent Application} \\
  {Permitted argument dependency} \\
  {in the return type} \\
\end{tabular}
                           & \begin{tabular}[c]{@{}l@{}}
                             Shallow \\
                             $(x: \ty[q]{T}) \to \ty[p]{S}$ \\
                             $x \notin \mathit{fv}(S)$
                             \end{tabular}
                           & \begin{tabular}[c]{@{}l@{}}
                             Deep (\secref{sec:motiv:depapp}) \\
                             $(x: \ty[q]{T}) \to \ty[p]{S}$ \\
                             $x \in \mathit{fv}(S)$ if $\vardiamondsuit \notin q$
                             \end{tabular} \\[3.5ex] %
\begin{tabular}[c]{@{}l@{}}
  \textbf{Type Abstraction} \\
  {Quantification over types}\\
\end{tabular}
                            & \begin{tabular}[c]{@{}l@{}}--\\ \ \end{tabular}
                            & \begin{tabular}[c]{@{}l@{}}
                              Bounded abstraction \`a la \Fsub \\
                              $\forall X <: T. \ty[p]{S}$ (\secref{sec2:typeabs})
                              \end{tabular}\\[2ex] %
\begin{tabular}[c]{@{}l@{}}
  \textbf{Reachability Abstraction} \\
  {Quantification over reachability}\\
\end{tabular}
                            & \begin{tabular}[c]{@{}l@{}}--\\ \ \end{tabular}
                            & \begin{tabular}[c]{@{}l@{}}
                              Bounded abstraction \`a la \Fsub\\
                              $\forall \ty[x]{X} <: \ty[q]{T}. \ty[p]{S}$ (\secref{sec2:qualabs})
                              \end{tabular} \\[2ex] %
\begin{tabular}[c]{@{}l@{}}
  \textbf{Mutable References}\\
  {Values stored in references}\\
\end{tabular}
                           & \begin{tabular}[c]{@{}l@{}}
                              Only flat \& untracked \\
                             $\textsf{Ref}[\ty[\bot]{T}]$
                             \end{tabular}
                           & \begin{tabular}[c]{@{}l@{}}
                              Possibly nested \& tracked \\
                             $\textsf{Ref}[\ty[q]{T}]$
                             (\secref{sec:nested-ref})
                             \end{tabular} \\
\bottomrule
\end{tabular}
\end{table}

We start by reviewing the reachability type system $\lambda^*$
\cite{DBLP:journals/pacmpl/BaoWBJHR21}, its limitations with regards to
reachability polymorphism, and then discuss our solution to this problem.
Our work shares a large proportion of the surface-language syntax with
\citet{DBLP:journals/pacmpl/BaoWBJHR21},  but differs in typing and semantics.
In examples, we use {\color{magenta} magenta} for
\citeauthor{DBLP:journals/pacmpl/BaoWBJHR21}'s type system and
{\color{dark-cyan} blue} for ours.

\Cref{tab:compare} summarizes the key
differences between the two systems and highlights the main improvements made in
this paper. First-time readers may safely skip this table.

\subsection{Revisiting $\lambda^*$ and Its Limitations} \label{sec:revisit}

\subsubsection{Reachability Sets as Qualifiers}
The $\lambda^*$ type system annotates types with reachability sets as
qualifiers, tracking the variables in the current environment that may be
reached by following memory references from the result of an expression.  For
example, consider an \baoterm{alloc()} function that yields a new resource of
fixed type \baotype{T}{} (\eg, a file handle). The qualifier of the result is
the empty set \baoterm{alloc():} \baotype{T}{\track}, since as a \emph{fresh}
value, it cannot reach any variables in the current environment.  When bound to
a variable \baoterm{x}, an invocation of \baoterm{alloc()} is not considered
fresh anymore as \baoterm{x} reaches \baoterm{x} itself:
\begin{lstlisting}[language=OOPSLA21]
  val x = alloc()  // : T$\trackset{x}$
\end{lstlisting}
Similarly, assigning a variable to another variable propagates reachability
by growing the set:
\begin{lstlisting}[language=OOPSLA21]
  val y = x        // : T$\trackset{x,y}$
\end{lstlisting}
Reachability is not symmetric, and stronger than aliasing, \eg, @y@ reaches
@x@, but @x@ does not reach @y@. It is cheaper to compute than full aliasing
and yet sufficient to check separation (\Cref{sec2:saturation}).

\subsubsection{Non-tracking Qualifier}
The $\lambda^*$ system assigns the bottom qualifier \bao{\(\bot\)} (often omitted) to untracked values. These usually
include base types, \eg, \lstinline[language=OOPSLA21]{42: Int$\untrack$}.
Untracked values can be treated as tracked by subtyping, but not vice versa.

\subsubsection{Function Types and Observable Separation}\label{sec:motiv:funtype}
The reachability qualifier of a function tracks its free variables and transitively their
implied reachability sets. For instance:
\begin{lstlisting}[language=OOPSLA21]
  val c = ...                  // : Ref[Int]$\trackset{c,x,y}$
  def f() = !c                 // : (f() => Int)$\trackset{f,c,x,y}$ $\leftarrow$ captures c and its qualifier
\end{lstlisting}
A function type has a self-reference (\eg, @f@ above, often omitted), which can be
used as an upper bound of its captured reachability set, \ie, it holds that \lstinline[language=OOPSLA21]|{c,x,y} <: {f}|,
but only inside the function body.
This is a mechanism for typing escaping closures (cf.~\Cref{sec:subtyping}), \eg,
when a function capturing @c@ escapes its defining scope, we can abstract
over the now free variable as follows:
\begin{lstlisting}[language=OOPSLA21]
  { val c = ...; { () => c } } // : (f() => Ref[Int]$\trackset{f}$)$\track$
\end{lstlisting}
Function return types also track reachability and
may mention other variables in the context, indicating possible aliases.
Argument qualifiers indicate the permissible overlap between a call-site
argument and the function's reachable set. Consider the following identity function:
\begin{lstlisting}[language=OOPSLA21]
  def id(x: T$\track$): T$\trackset{x}$ = x         $\,\,\,$// : ((x: T$\track$) => T$\trackset{x}$)$\track$
\end{lstlisting}
Its argument qualifier is the empty set, demanding that the argument cannot
be aliased with the free variables of the function.
This is the \emph{observable separation guarantee} of reachability types.

\subsubsection{Reachability Polymorphism and its Limitations} \label{sec:motiv:limit}

$\lambda^*$ provides a \emph{lightweight} form of reachability
polymorphism via dependent function applications, \eg,
consider the @id@ function from above:
\begin{lstlisting}[language=OOPSLA21]
  val x: T$\trackset{x,a,b}$ = ...; id(x)        // : T$\color{magenta}^{\tt{\{x\}[x \mapsto \{x, a, b\}]}}$ = T$\tracksetl{x, a, b}$
  val y: T$\tracksetl{y, z}\ \ \,\,$ = ...; id(y)        // : T$\color{magenta}^{\tt{\{x\}[x \mapsto \{y, z\}]}}\ $  = T$\tracksetl{y, z}$
\end{lstlisting}
The type of @id@ mentions no explicit quantifiers, and yet can be regarded as
polymorphic over a fixed base type @T@ with any reachability qualifier \(q\),
as long as \(q\) is disjoint from @id@'s reachability set. Since @id@ itself
has an empty qualifier, any \(q\) is acceptable.

\paragraph{The Root of the Problem: Confusing Untracked with Fresh Values}
The problem with reachability polymorphism in \(\lambda^*\) is its non-parametric
treatment of untracked versus tracked arguments, \eg,
the \baoterm{id} function conflates these two different instantiations:
\begin{lstlisting}[language=OOPSLA21]
  val z = ...                  // : T$\untrackl$
  id(z)                        // : T$\color{magenta}^{\tt{\{x\}[x \mapsto \bot]}}$ = T$\track$ $\leftarrow$ untracked value now considered tracked
  id(alloc())                  // : T$\color{magenta}^{\tt{\{x\}[x \mapsto \varnothing]}}$ = T$\track$
\end{lstlisting}
Qualifier substitution with untracked status yields
$\color{magenta}{\tt{\{x\}[x \mapsto \bot]} = \varnothing}$
a tracked qualifier without known aliases (\ie, fresh).
\citeauthor{DBLP:journals/pacmpl/BaoWBJHR21} (Section 3.4)
made this design choice to ensure soundness,
but it introduces imprecision in tracking status and constitutes a severe
limitation in expressiveness.  No code path can be generic with respect to the
tracking status of arguments!
To see why admitting a more precise qualifier \baotype{T}{\untrack} for @id(z)@
is unsound, we can postulate this ``more precise'' behavior (\ie, assuming
$\color{magenta}{\tt{\{x\}[x \mapsto \bot]} = \bot}$) and subvert the
type system.
Consider the function  @fakeid@ returning a fresh tracked value each time:
\begin{lstlisting}[language=OOPSLA21]
  def fakeid(x: T$\track$): T$\trackset{x}$ = alloc()
\end{lstlisting}
This function typechecks since the body expression has type
\baotype{T}{\track}, which is a subtype of the declared return type
\baotype{T}{\trackset{x}}.
Under the postulate, applying @fakeid@ with a non-tracking arguments results in
\begin{lstlisting}[language=OOPSLA21]
  val y = ...                  // : T$\untrackl$
  fakeid(y)                    // : T$\color{magenta}^{\tt{\{x\}[x \mapsto \bot]}}$ = T$\untrackl$ $\leftarrow$ $\color{red}\texttt{unsound!}$
\end{lstlisting}
But @fakeid(y)@ actually returns a fresh value of qualifier $\color{magenta}\varnothing$
that should never be down-cast to untracked!
This violates the \emph{separation guarantee} of the type system: a tracked
value cannot escape as an untracked value. Otherwise, it can no longer be kept
separate from other tracked values.

To summarize, reachability polymorphism via dependent application in $\lambda^*$ must
sacrifice parametricity and  precision for soundness, leading to a confusion of untracked with fresh values.
There is no easy fix with the binary track/untrack distinction, and we must rethink reachability polymorphism
and the notion of freshness.

\subsection{Precise Reachability Polymorphism in \maybelang}\label{sec:motiv-reach-poly}

We propose the $\maybelang$-calculus, which features a new treatment of
freshness and a finer-grained reachability assignment, leading to a well-behaved
and more precise notion of reachability polymorphism that smoothly scales to type-and-qualifier
abstraction.

\subsubsection{One-Step Reachability Tracking}\label{sec2:ondemand}

\citet{DBLP:journals/pacmpl/BaoWBJHR21} use an \emph{``eager''} strategy to track aliases:
typing relations assign \emph{saturated} qualifiers, \ie, these qualifiers are large enough
to include all transitively reachable variables.  In contrast, $\maybelang$
keeps reachability sets minimal in type assignment and only computes transitive
closures \emph{on demand} (cf.~\Cref{sec2:saturation}), which ensures that we can preserve
chains of reachability and refine elements in the chain later by substitution or subtyping.

The {\color{magenta} ``eager''} and {\color{dark-cyan}``on-demand''} tracking strategies
each treat variable bindings differently (typing context shown to the right of $\dashv$):

\begin{minipage}{0.49\linewidth}
\begin{lstlisting}[language=OOPSLA21]
val x = alloc()  // : T$\trackset{x}$
val y = x        // : T$\trackset{x, y}$
\end{lstlisting}
\end{minipage}%
\begin{minipage}{0.49\linewidth}
\begin{lstlisting}
val x = alloc()  // : T$\trackset{x}$ $\dashv$ x: T$\trackset{\qfresh}$
val y = x        // : T$\trackset{y}$ $\dashv$ y: T$\trackset{x}$, x: T$\trackset{\qfresh}$
\end{lstlisting}
\end{minipage}\\
In the eager version (left),  @y@ reaches {\color{magenta}@{x, y}@},
transitively including @x@'s reachability set from the context.  The on-demand
version (right) only assigns the one-step reachability set
{\color{dark-cyan}@{y}@}. It can be scaled to the saturated set by subtyping
({\color{dark-cyan}@{y} <: {x,y}@}) which includes the subset relation.
On-demand tracking preserves the chains of reachability in typing: during
reduction steps, qualifiers in the chain can be replaced with smaller reachable
sets, leading to an increase in precision via substitution.

\subsubsection{Freshness Marker $\QFresh$}\label{sec:freshness}
We model potential freshness by adding a marker \(\QFresh\) to qualifiers,
connecting static observability with evaluation.  Consistent with observable
separation (\Cref{sec:motiv:funtype}), a type $\ty[\{\qfresh\}]{T}$ describes
expressions which cannot reach the currently observable variables, but they may
reach unobservable variables, including new references.  The prime example is
the reduction of allocations:
\begin{lstlisting}
  alloc() // : Ref[Int]$\trackset{\qfresh}$ $\quad\color{black}\longrightarrow\quad$ $\color{black}\ell$   // : Ref[Int]$^{\text{\texttt{\{\(\ell\)\}}}}$, where $\ell$ is a fresh location value
\end{lstlisting}
Before reduction, @alloc()@ is fresh, \ie, it must be tracked but is not bound
to a variable.  Afterwards, we have a new and definitely known store location,
which is considered not fresh, thus \(\qfresh\) vanishes.  The presence of
\(\qfresh\) indicates that reduction steps may grow the qualifier, and its
absence indicates that they will not.
\citeauthor{DBLP:journals/pacmpl/BaoWBJHR21}'s track/untrack system assumes
that any tracked qualifier might grow.

The \(\qfresh\) marker also serves as a ``contextual freshness'' indicator for
function parameters, \eg, here is the reachability-polymorphic identity
function in \maybelang:
\begin{lstlisting}
  def id(x: T$\trackfresh$): T$\trackset{x}$ = x                            // : ((x: T$\trackfresh$) => T$\trackset{x}$)$\track$
\end{lstlisting}
The type specifies that @id@ (1) cannot observe anything about its context
({\color{dark-cyan}\(\varnothing\)}), and (2) it accepts arguments that may
reach any unobservable variables. Thus, the @id@ function  accepts @T@
arguments with any qualifier and the function body can only observe a fresh
argument.  Adjusting parameter qualifiers permits  controlling the overlap
between functions and their arguments, \eg, consider variants of @id@ which
close over some variable @z@ in context:
\begin{lstlisting}
  def id2(x: T$\trackset{\qfresh}\,\,\,\,\,$): T$\trackset{x}$ = { val u = z; x }            // : ((x: T$\trackset{\qfresh}\,\,\,\,\,$) => T$\trackset{x}$)$\trackset{z}$
  def id3(x: T$\trackset{\qfresh,z}$): T$\trackset{x}$ = { val u = z; x }            // : ((x: T$\trackset{\qfresh,z}$) => T$\trackset{x}$)$\trackset{z}$
  def id4(x: T$\trackset{z}\,\,\,\,\,$): T$\trackset{x}$ = { val u = z; x }            // : ((x: T$\trackset{z}\,\,\,\,\,$) => T$\trackset{x}$)$\trackset{z}$
\end{lstlisting}
The qualifiers on the function type and the parameter specify the reachability
information that the implementation can \emph{observe} about its context (only
@z@ here), and about any given argument, respectively. It also says that the
implementation is \emph{oblivious} to anything it \emph{cannot observe}.
Function @id2@ accepts arguments reaching anything that does not (directly or
transitively) reach @z@.  But @id3@ permits @z@ in the argument's qualifier,
effectively allowing any argument. Finally, @id4@'s parameter lacks the
freshness marker, constraining arguments to be contextually non-fresh. That is,
only observable arguments which reach at most @z@ are allowed.

With the freshness marker, it is no longer necessary to use
{\color{magenta}$\bot$} to indicate untracked values.
In \maybelang, qualifiers of untracked values (\eg, primitive values) are simply
denoted by the empty set {\color{dark-cyan}$\varnothing$}.

\subsubsection{Precise Reachability Polymorphism}\label{sec:motiv:precise}

Unlike its \(\lambda^*\) counterpart, @id@ is truly reachability polymorphic,
because it properly preserves the tracking status of arguments:
\begin{lstlisting}
  id(42)                       // : Int$\color{dark-cyan}^{\tt{\{x\}[x \mapsto \varnothing]}}$ = Int$\trackl\,$   $\leftarrow$ unbound and untracked
  id(alloc())                  // : T$\color{dark-cyan}^{\tt{\{x\}[x \mapsto \QFresh]}}$   = T$^{\{\QFresh\}}\ \ \ \,$     $\leftarrow$ unbound and tracked (fresh)
\end{lstlisting}
The key design difference here is having the \(\qfresh\) marker in qualifiers
to explicitly communicate (non-)freshness which is preserved by dependent
application and substitution.
Consider a function that mutates a captured reference cell and returns the
argument.  We annotate that the argument @x@ is potentially aliased with the
captured argument @c1@.  However, in \citet{DBLP:journals/pacmpl/BaoWBJHR21}'s
system, this potential alias is propagated to the return type qualifier and we
cannot get rid of it even when applying with a non-overlapped argument @c2@:
\begin{lstlisting}[language=OOPSLA21]
  ... // c1: T$\trackset{c1}$, c2: T$\trackset{c1}$
  def foo(x: T$\trackset{c1}$): T$\trackset{c1,x}$ = { c1 := !c1 + 1; x } // : ((x: T$\trackset{c1}$) => T$\trackset{c1,x}$)$\trackset{c1}$
  foo(c1)                                         // : T$\trackset{c1}$
  foo(c2)                                         // : T$\trackset{c1, c2}$  $\leftarrow$ imprecise!
\end{lstlisting}
In contrast, $\maybelang$ would not propagate such imprecision by tracking
one-step reachability. The return type only tracks the argument @x@.
When applying different arguments to the function, precise reachability
is retained:
\begin{lstlisting}
  ... // c1: T$\trackset{c1}$, c2: T$\trackset{c1}$
  def foo(x: T$\trackset{c1,\qfresh}$): T$\trackset{x}$ = { c1 := !c1 + 1; x }     // : ((x: T$\trackset{c1,\qfresh}$) => T$\trackset{x}$)$\trackset{c1}$
  foo(c1)                                        // : T$\trackset{c1}$
  foo(c2)                                        // : T$\trackset{c2}$  $\leftarrow$ precision retained
\end{lstlisting}

The freshness marker also prevents typing the problematic @fakeid@ function, since $\{\QFresh\}$
is not compatible with the result qualifier {\color{dark-cyan}@{x}@}:
\begin{lstlisting}
  def fakeid(x: T$\trackfresh$): T$\trackset{x}$ = alloc()                  // $\color{red}\texttt{type error}$: $\text{\texttt{\{\qfresh\}}}\not{<:}\text{\texttt{\{x\}}}$
\end{lstlisting}

\subsubsection{Maybe-Tracked and Subtyping} \label{sec:motiv:maybe}
With the one-step reachability tracking, we introduce a novel notion of
``maybe-tracked'' status in $\maybelang$.  For example, the tracking status of
@Int$\trackset{x}$@ only depends on the reachability of @x@, and therefore is
``maybe'' tracked:
\begin{lstlisting}
  val x = 42           // : Int$\trackset{x}$            $\leftarrow$ bound but untracked
  id(x)                // : Int$\trackset{x}$ <: Int$\track$     $\leftarrow$ unbound and upcast via one-step reachability
\end{lstlisting}
Moreover, @Int$\trackset{x}$@ is equivalent to @Int$\track$@, upcast by one-step
reachability using $\maybelang$'s subtyping relation.
Chasing the typing assumptions, both {\color{dark-cyan}@{x} <: $\varnothing$@}
and {\color{dark-cyan}@$\varnothing$ <: {x}@} hold in the above context, which
justifies the equivalence.
This reasoning step uses a subtyping rule for looking up qualifiers of bound
variables in the context (see \Cref{sec:subtyping}), which permits smaller,
context-dependent steps to form reachability chains \emph{as long as qualifiers
in the chain are all non-fresh}. Therefore, @id(y)@ cannot be upcast since
its one-step reachable variable @y@ is fresh:
\begin{lstlisting}
  val y = alloc()      // : T$\trackset{y}$              $\leftarrow$ bound and tracked
  id(y)                // : T$\trackset{y}$              $\leftarrow$ bound and cannot further upcast since y fresh
\end{lstlisting}

\subsubsection{On-Demand Transitivity}\label{sec2:saturation}

When does the type system actually need  to compute saturated qualifiers
with the ``on-demand'' tracking strategy (\Cref{sec2:ondemand})?
Applying functions that expect fresh arguments is the only situation where this is necessary.
For example, consider a function @f@ that does not permit overlap
between the argument's qualifier and its own reachable set:
\begin{lstlisting}
  val c1 = alloc()                   // : Ref[Int]$\trackset{c1}$ $\dashv$ c1: Ref[Int]$\trackfresh$
  def f(x: Ref[Int]$\trackfresh$) = !c1 + !x      // : (f(x: Ref[Int]$\trackfresh$) => Int)$\trackset{c1}$
  val c2 = c1                        // : Ref[Int]$\trackset{c2}$
  f(c2)                              // $\color{red}\texttt{type error}$: since {c1,c2} $\cap$ {c1} $\neq \varnothing$
\end{lstlisting}
The application @f(c2)@ should  be rejected due to the lack of separation
between @c2@ and @f@. Since the one-step reachability strategy lets variable bindings
reach only themselves by default, naively intersecting the function and argument at the call site
would not detect that @c2@ overlaps with @f@ through @c1@. Thus,
a sound overlap check at call sites must first compute saturated upper bounds on demand,
and then compute their intersection.
We discuss the formal details of  saturated qualifiers and overlap checking further in \Cref{sec:laziness}.

Finally, it is worth noting that @c2@'s qualifier cannot be upcast through its
reachability chain @c1@ to \(\{\QFresh\}\) via subtyping,
which would result in unsound overlap checking (cf.~\Cref{sec:subtyping}).

\subsubsection{Qualifier-Dependent Application} \label{sec:motiv:depapp}

Recall that the $\lambda^*$-calculus achieves reachability polymorphism
via dependent function application (\Cref{sec:motiv:limit}). That is, given a function type
\(f(x: \ty[q_1]{T_1}) \to \ty[q_2]{T_2}\), both \(x\) and \(f\) may occur in the
codomain qualifier \(q_2\), but the system forbids occurrences within \(T_2\)
to ensure a sound treatment of escaping closures~\cite{DBLP:journals/pacmpl/BaoWBJHR21}.
Therefore, only shallow dependencies are allowed in applications.

The root cause is that all tracked qualifiers in $\lambda^*$ can potentially
grow with unobservable reachability sets. Due to \maybelang{}'s refined
freshness-marker model, we can distinguish fresh/growing from non-fresh/static
qualifiers, and safely permit occurrence of \(f\) and \(x\) deeply in \(T_2\)
in the latter case (cf. \Cref{sec:depend-appl}) without precision loss.
Consider the following function returning another function:
\begin{lstlisting}
  val c = alloc()
  def f(x: Ref[Int]$\trackset{c}$) = () => x // : (f(x: Ref[Int]$\trackset{c}$) => (Unit => Ref[Int]$\trackset{x}$)$\trackset{x}$)$\track$
  f(c)                           $\,$// : (Unit => Ref[Int]$\trackset{c}$)$\trackset{c}$
\end{lstlisting}
We can assign the reachability set of @f@'s innermost return type, depending
on the outer argument @x@.
The dependent application @f(c)@ yields a precise type,
whereas the $\lambda^*$-calculus would have to
upcast the returned function type to a self-reference before application (thus introducing imprecision).

\subsection{Type-and-Qualifier Abstractions in \polylang}
\label{sec2:qualabs}

Because of its confounding of fresh tracked values and untracked values, the
$\lambda^*$-calculus lacks type abstraction mechanisms such as generic types.
In contrast, we can smoothly extend \maybelang with type-and-qualifier abstractions
in the style of $F_{<:}$ \cite{DBLP:journals/iandc/CardelliMMS94}.

\subsubsection*{Type Abstractions} \label{sec2:typeabs}

The first step towards \polylang is to add \Fsub-style quantification
over proper types \emph{without} qualifiers.
This is already attractive and enough to express the identity function with
\emph{both} type and lightweight reachability polymorphism.
The following definition of @id@ adds the type parameter @T@ and does not require
\Fsub-style abstraction of qualifiers:
\begin{lstlisting}
  def id[T $<:$ Top](x: T$\trackfresh$): T$\trackset{x}$ = x
\end{lstlisting}
As in \Fsub, we add an upper bound @Top@ of all types to the
system. However, reachability sets attached to proper types must be
concrete and cannot be abstracted over.

\subsubsection*{Qualifier Abstractions}

We now introduce an \emph{abstract qualifier} and an \emph{upper bound
qualifier} in the style of \Fsub.
In this way, the polymorphic identity function is a shorthand notation that
does not need to use the abstract qualifier. The fully desugared term is
\begin{lstlisting}
  def id[T$\trackvar{z}$ <: Top$\trackfresh$](x: T$\trackfresh$): T$\trackset{x}$ = x
  def id[T](x: T$\trackfresh$) = x // shorthand notation
\end{lstlisting}
where @z@ is the abstract qualifier variable bounded by $\QFresh$.
One could further omit the abstract qualifier, type-and-qualifier bound,
and return type using the shorthand notation shown above.

Although the additionally introduced abstract qualifier
({\color{dark-cyan}@z@}) does not yield further expressiveness for the identity
function, quantified qualifiers vary independently of the type variable, and
one is free to attach them to any proper type.
In \Cref{sec:poly}, we present the formalization of \polylang, which combines
\maybelang with  \Fsub-style polymorphism for bounded type-and-qualifier
abstraction.

\subsection{Polymorphic Data Types} \label{sec2:polydata}
In this section, we consider typing polymorphic data types under \polylang and
demonstrate the expressiveness gain from type-and-qualifier polymorphism.
Suppose we have extended the language with native pair types, how should their
typing rules look like? There are two main design goals:
\begin{itemize}[leftmargin=2em]
  \item First, we would like to precisely track the reachability of components,
    so a pair type @Pair[A$\trackvar{a}$, B$\trackvar{b}$]@ annotates qualifiers to components.
    Moreover, the projection functions should preserve precise reachability
    whenever possible.
    For example, given an expression of type @Pair[A$\trackvar{a}$, B$\trackvar{b}$]@,
    retrieving its components should yield exactly the same qualifiers we put in:
\begin{lstlisting}
...                // u: Ref[Int]$\trackvar{u}$, v: Ref[Int]$\trackvar{v}$
val p = Pair(u, v) // : Pair[Ref[Int]$\trackvar{u}$, Ref[Int]$\trackvar{v}$]$\trackvar{p}$
fst(p)             // : Ref[Int]$\trackvar{u}$                      $\leftarrow$ precision retained
snd(p)             // : Ref[Int]$\trackvar{v}$                      $\leftarrow$ precision retained
\end{lstlisting}
    The above snippet creates a pair of two reference cells and then gets its components.
    Explicit type applications are omitted and can be inferred as in Scala
    (\eg, by bidirectional typing \cite{DBLP:journals/toplas/PierceT00}).

  \item Second, we would like to allow pairs capturing local variables
    to escape from their defining scope (\eg, the counter example in
    \Cref{fig:counter}). To this end, we designate a self-reference @p@
    for pairs $\mu$@p.Pair[A$\trackvar{a}$, B$\trackvar{b}$]@,
    which serves as an upper bound of the pair-component reachability.
    To handle escaped pairs, the key insight is similar to function types: we
    can replace arbitrary component qualifiers (as they are in covariant
    positions) with self references via subtyping
    @Pair[A$\trackvar{a}$, B$\trackvar{b}$]$\trackvar{q}$@ $<:$ @(@$\mu$@p.Pair[A$\trackvar{p}$, B$\trackvar{p}$])$\trackvar{q}$@,
    just as with function subtyping where the codomain's qualifier can be upcast to
    function's self-reference:
    \begin{lstlisting}
def f() = {
  ...              // u: Ref[Int]$\trackvar{u}$, v: Ref[Int]$\trackvar{v}$
  Pair(u, v)       // : Pair[Ref[Int]$\trackvar{u}$, Ref[Int]$\trackvar{v}$]$\trackset{u,v}$
}                  //   upcast to $\mu$p.Pair[Ref[Int]$\trackvar{p}$, Ref[Int]$\trackvar{p}$]$\track$ when escaping
    \end{lstlisting}
    Once the pair is bound to a variable, we ``unpack'' the self-reference so
    that projections are properly aliased.
    \begin{lstlisting}
val p = f()        // now u and v are not in the context:
fst(p)             // : Ref[Int]$\trackvar{p}$
snd(p)             // : Ref[Int]$\trackvar{p}$
    \end{lstlisting}
\end{itemize}
Aiming for minimality, the rest of this section investigates the typing of
Church-encoded pairs that satisfies our desired typing and subtyping rules.
We discuss two different types of encodings: ``transparent'' and ``opaque'' pairs
corresponding to the two usage scenarios above.
Transparent pairs track precise reachability of components using \polylang's
parametric qualifiers and can only be used under appropriate contexts.  Opaque
pairs use self-references as an abstraction to hide local
qualifiers and can escape to an outer scope.  Finally, the subtyping rule
connecting both is justified by a coercion function that eta-expands pairs,
converting transparent pairs to opaque pairs.

\subsubsection{Typing Church Pairs, Transparently}\label{para:church-pair}

The transparent pair type @Pair[A,B]@ is defined as a universal type with
argument type @C@.  We also introduce abstract qualifiers for type @A@, @B@,
and @C@.  Moreover, the qualifier of result type @C@ is simply parametric.
\begin{lstlisting}
  type Pair[A$\trackvar{a}$ <: Top$\trackfresh$, B$\trackvar{b}$ <: Top$\trackset{a,\qfresh}$] =
    [C$\trackvar{c}$ <: Top$\trackset{a,b,\qfresh}$] => ((A$\trackvar{a}$, B$\trackvar{b}$) => C$\trackvar{c}$)$\track$ => C$\trackvar{c}$)$\trackset{c,a,b}$
\end{lstlisting}
We also assume the base system is extended with multi-argument
functions (instead of currying arguments), where each argument is disjoint from
others.
Similarly, the term constructor uses @C@'s qualifier for the application @f(a, b)@:
\begin{lstlisting}
  def Pair[A$\trackvar{a}$ <: Top$\trackfresh$, B$\trackvar{b}$ <: Top$\trackset{a, \qfresh}$](a: A$\trackvar{a}$, b: B$\trackvar{b}$): Pair[A, B]$\trackset{a,b,\qfresh}$ =
    [C$\trackvar{c}$ <: Top$\trackset{a,b,\qfresh}$] => (f: (A, B) => C) => f(a, b)
\end{lstlisting}
When using the quantified type for the argument or return type, its
accompanying qualifier is implicitly attached, \ie, we write
@A@ as a shorthand of @A$\trackvar{a}$@ when using it.

The projectors @fst@ and @snd@ have the usual definition but using
accurate types and qualifiers:
\begin{lstlisting}
  def fst[A$\trackvar{a}$ <: Top$\trackfresh$, B$\trackvar{b}$ <: Top$\trackset{a, \qfresh}$](p: Pair[A$\trackvar{a}$, B$\trackvar{b}$]$\trackset{a,b,\qfresh}$): A$\trackvar{a}$ = p((a, b) => a)
  def snd[A$\trackvar{a}$ <: Top$\trackfresh$, B$\trackvar{b}$ <: Top$\trackset{a, \qfresh}$](p: Pair[A$\trackvar{a}$, B$\trackvar{b}$]$\trackset{a,b,\qfresh}$): B$\trackvar{b}$ = p((a, b) => b)
\end{lstlisting}
By making the elimination type @C@'s qualifier parametric, we can now
instantiate it in the projection function with the precise component
qualifiers, as shown by the example at the beginning of this section.
The general Church-encoding of data types via sums and products can also
benefit from the increased precision.

\subsubsection{Typing Escaped Church Pairs, Opaquely}
The transparent pair typing works for cases where the components are still in the
context, but the pair cannot escape from that scope (cf. \Cref{fig:counter}).
We now discuss the types of escaped pairs using \emph{self-references} as abstraction.
To avoid confusion, we name the type and constructor of opaque pairs as @OPair@,
and transparent pairs remain @Pair@.
\begin{lstlisting}
  type $\color{black}{\mu}$p.OPair[A$\trackvar{a}$ <: Top$\trackfresh$, B$\trackvar{b}$ <: Top$\trackset{a,\qfresh}$] = // p: self-reference of a pair instance
    p[C$\trackvar{c}$ <: Top$\track$] => (h((x: A$\trackfresh$, y: B$\trackset{x,\qfresh}$) => C$\trackset{x,y}$) => C$\trackvar{h}$)$\trackvar{p}$
\end{lstlisting}
Note that in $\polylang$ universal types and type abstractions also have
self-references (\eg, @p@ in the definition) that can be used
to express escaping polymorphic closures, similar to their
term-level correspondences (\eg, @h@ in the definition).
Therefore, the self-reference in @$\color{black}{\mu}$p.OPair@ is just
an syntactic annotation referring to the self-reference of the universal type.
Compare to the transparent typing, here we do not use quantified qualifiers that are
parametrically introduced.
Instead, we use a chain of self-references
in the codomains, upcasting from the inner most reachability @{x, y}@ to @h@
and to @p@.
The introduction and elimination forms of opaque pairs also reflect
the typing using self-references:
\begin{lstlisting}
  def OPair[A$\trackvar{a}$ <: Top$\trackfresh$, B$\trackvar{b}$ <: Top$\trackset{a,\qfresh}$](a: A, b: B): $\color{black}{\mu}$p.OPair[A, B]$\trackset{a,b}$ =
    [C$\trackvar{c}$ <: Top$\track$] => (f: (x: A$\trackfresh$, y: B$\trackset{x,\qfresh}$) => C$\trackset{x,y}$) => f(a, b)
  def fst[A$\trackvar{a}$ <: Top$\trackfresh$, B$\trackvar{b}$ <: Top$\trackset{a,\qfresh}$](p: $\color{black}{\mu}$p.OPair[A, B]$\trackset{a,b}$): A$\trackvar{p}$ = p((a, b) => a)
  def snd[A$\trackvar{a}$ <: Top$\trackfresh$, B$\trackvar{b}$ <: Top$\trackset{a,\qfresh}$](p: $\color{black}{\mu}$p.OPair[A, B]$\trackset{a,b}$): B$\trackvar{p}$ = p((a, b) => b)
\end{lstlisting}

\subsubsection*{Imprecise Eliminations}
While the typing works out, the resulting qualifiers of the projections
@fst@/@snd@ are imprecise.
We have no means to vary the qualifier of the elimination type @C@ in type
@OPair@.
When the component qualifiers are not available in the context,
using the self-reference to track possible sharing is the most accurate
option.
This is the intended design as shown in the beginning of this section.
A side effect of such typing is that in-scope elimination can yield the set of
joint qualifiers, since the pair can reach them by our ``maybe-tracked''
notation:
\begin{lstlisting}
  ...                 // u and v defined as before
  val p = OPair(u, v) // : $\mu$p.OPair[Ref[Int], Ref[Int]]$\track$ binds to p, unpacking the self-ref
  fst(p)              // : Ref[Int]$\trackvar{p}$ <: Ref[Int]$\trackset{u,v}$  $\leftarrow$ imprecise joint qualifiers
  snd(p)              // : Ref[Int]$\trackvar{p}$ <: Ref[Int]$\trackset{u,v}$  $\leftarrow$ imprecise joint qualifiers
\end{lstlisting}

\subsubsection*{Conversion between Opaque and Transparent Pairs}
The two different types for Church-encoded pairs are connected, \ie,
transparent pairs can be converted to opaque via eta-conversion:
\begin{lstlisting}
  def conv[A$\trackvar{a}$ <: Top$\trackfresh$, B$\trackvar{b}$ <: Top$\trackset{a,\qfresh}$](p: Pair[A, B]$\trackset{a,b,\qfresh}$): $\color{black}{\mu}$p.OPair[A, B]$\trackset{a,b}$ =
    OPair(fst(p), snd(p))
\end{lstlisting}
From a pragmatic perspective, when the language is extended with pairs as native
algebraic data types, the eta-conversion justifies an admissible subtyping
rule for escaped pairs.

It is important to note that both the transparent and opaque pair encodings
use the same \emph{terms}, namely the standard System-F encoding,
just with different assigned \emph{type qualifiers}.
We expect that the core \polylang typing and subtyping rules can be refined
or extended to enable a uniform encoding so that the eta-conversion step
becomes unnecessary.

\subsection{Nested Mutable References} \label{sec:nested-ref}

The base type system of \citeauthor{DBLP:journals/pacmpl/BaoWBJHR21} supports
reference cells that can only store ``untracked'' values or pure computation
(\ie, of qualifier {\color{magenta}$\bot$}).
This compromise is again due to conflating untracked and fresh values in
\(\lambda^*\).
To support more expressive nested references, the $\lambda^*$-calculus
has to use a flow-sensitive effect system with explicit move semantics.
This is not at all required in the  \(\maybelang\)-calculus, because its
freshness model already  supports a form of nested references.

The key idea is that a reference's content also carries a reachability annotation, \eg,
@Ref[T$\, \trackvar{p}$]$\trackvar{q}$@, where \(\qfresh\notin\text{{\color{dark-cyan}\lstinline{p}}}\).
That is, only references with fully observable reachability  are permitted, and
these references remain invariant once introduced, and can only be assigned
with values having the same reachability set.

This restrictive model for \maybelang can already express more interesting programs than \(\lambda^*\).
For example, we could define complex heap-allocated data
structures or store effectful functions into reference cells.
Recall the @counter@ example (\Cref{fig:counter}) that returns two functions to
increase or decrease an encapsulated state\footnote{
  Because of the subtyping discussed in \Cref{sec2:polydata}, we do not
  distinguish transparent and opaque pairs and assume the cast is implicitly
  applied when necessary.
}. Both functions share the
same reachable set containing @ctr@:
\begin{lstlisting}
  val ctr = counter(0)   // : Pair[(()=>Unit)$\tracksetl{ctr}$,(()=>Unit)$\tracksetl{ctr}$]$\tracksetl{ctr}$
  val incr = fst(ctr)    // : (()=>Unit)$\tracksetl{ctr}$
  val decr = snd(ctr)    // : (()=>Unit)$\tracksetl{ctr}$
\end{lstlisting}
Note that @incr@ and @decr@ encapsulate and mutate a locally-defined
heap reference cell, thus are effectful.
We could create a reference cell that stores either the @incr@ or @decr@ function:
\begin{lstlisting}
  val cf = new Ref(incr) // : Ref[(()=>Unit)$\tracksetl{ctr}$]$\trackset{cf}$
  cf := decr             // : Unit$\track$
\end{lstlisting}

This pattern permits more flexible uses of these capabilities, \eg,
registering functions as callbacks or tracking permissible escaping
via assignments.
\Cref{sec:mutable-refs} discusses the formal rules of this restricted form of
nested references.

\section{Simply-Typed Reachability Polymorphism}\label{sec:maybe}

This section presents the formal metatheory of the base \maybelang-calculus
(\Cref{sec:motiv-reach-poly}), a generalization of the \(\lambda^*\)-calculus
by \citet{DBLP:journals/pacmpl/BaoWBJHR21} that adds the notion of freshness markers for a more
precise notion of lightweight qualifier polymorphism.

\begin{figure}[t]\small
\begin{mdframed}
\judgement{Syntax}{\BOX{\maybelang}}
  $$\begin{array}{l@{\qquad}l@{\qquad}ll}
    x,y,z   &\in & \Var                                                                            & \text{Variables} \\
    f,g,h   &\in & \Var                                                                            & \text{Function Variables} \\
    t       &::= & c \mid x \mid \lambda f(x).t \mid t~t \mid \tref~t \mid\ !~t \mid t \coloneqq t & \text{Terms}\\[2ex]
    p,q,r   &\in & \mathcal{P}_{\mathsf{fin}}(\Var \uplus \{ \QFresh \})                           & \text{Reachability Qualifiers} \\
    S,T,U,V &::= & B \mid f(x: \ty{Q}) \to \ty{Q} \mid \TRef~Q                                     & \text{Types} \\
    O,P,Q,R &::= & \ty[q]{T}                                                                       & \text{Qualified Types} \\[2ex]
    \flt    &\in & \mathcal{P}_{\mathsf{fin}}(\Var)                                                & \text{Observations} \\
    \Gamma  &::= & \varnothing\mid \Gamma, x : Q                                                   & \text{Typing Environments} \\
    \end{array}$$\\
\textbf{\textsf{Qualifier Notations}} \qquad
$\small p,q := p \qlub q\qquad x := \{x\} \qquad \QFresh :=\{\QFresh\}\qquad \starred{q} := \{\qfresh\}\qlub q $
\caption{The syntax of \maybelang.}
\label{fig:maybe:syntax}
\end{mdframed}
\vspace{-1em}
\end{figure}

\begin{figure}[t]\small
\begin{mdframed}
\judgement{Term Typing}{\BOX{\strut\G[\flt] \ts t : \ty{Q}}}\\[1ex]
\begin{minipage}[t]{.47\linewidth}\vspace{0pt}
    \infrule[t-var]{
      x : \ty[q]{T} \in \G\quad\quad x \in \flt
    }{
      \G[\flt] \ts x : \ty[x]{T}
    }
\vgap
    \infrule[t-abs]{
      \cx[q,x,f]{(\G\ ,\ f: \ty{F}\ ,\ x: \ty{P})} \ts t : \ty{Q}\quad q\subq \flt\\
      \ty{F} = \ty[q]{\left(f(x: \ty{P}) \to \ty{Q}\right)}
    }{
      \G[\flt] \ts \lambda f(x).t : \ty{F}
    }
\vgap
    \infrule[t-app]{
      \G[\flt]\ts t_1 : \ty[q]{\left(f(x: \ty[p]{T}) \to \ty{Q}\right)} \qquad
      \G[\flt]\ts t_2 : \ty[p]{T}\\\QFresh\notin p\qquad Q = \ty[r]{U}\qquad r\subq\starred{\varphi,x,f}
    }{
      \G[\flt]\ts t_1~t_2 : \ty{Q}[p/x, q/f]
    }
\vgap
    \infrule[t-app$\QFresh$]{
      \G[\flt]\ts t_1 : \ty[q]{\left(f(x: \ty[p\,{\overlap}\, q]{T}) \to Q\right)}\quad
      \G[\flt]\ts t_2 : \ty[p]{T}\\ Q = \ty[r]{U}\qquad r\subq\starred{\varphi,x,f}\\
      \QFresh \in p \Rightarrow x\notin\FV(U)\quad \QFresh \in q \Rightarrow f\notin\FV(U)
    }{
      \G[\flt]\ts t_1~t_2 : Q[p/x, q/f]
    }
\end{minipage}%
\begin{minipage}[t]{.03\linewidth}
\hspace{1pt}%
\end{minipage}%
\begin{minipage}[t]{.5\linewidth}\vspace{0pt}
    \infrule[t-cst]{
      c \in B
    }{
      \G[\flt] \ts c : \ty[\qbot]{B}
    }
\vgap
  \infrule[t-ref]{
      \G[\flt]\ts t : \ty[q]{T}\qquad \QFresh\notin q
    }{
      \G[\flt]\ts \tref~t : \ty[\starred{q}]{(\TRef~\ty[q]{T})}
    }
\vgap
    \infrule[t-deref]{
      \G[\flt]\ts t : \ty[q]{(\TRef~\ty[p]{T})}\qquad \QFresh\notin p\qquad p\subq\flt
    }{
      \G[\flt]\ts !t : \ty[p]{T}
    }
\vgap
    \infrule[t-assgn]{
      \G[\flt]\ts t_1 : \ty[q]{(\TRef~\ty[p]{T})} \quad
      \G[\flt]\ts t_2 : \ty[p]{T} \quad
      \QFresh\notin p
    }{
      \G[\flt]\ts t_1 \coloneqq t_2 : \ty[\qbot]{\TUnit}
    }
\vgap
    \infrule[t-sub]{
      \G[\flt]\ts t : \ty{Q} \quad  \G\ts\ty{Q} <: \ty[q]{T}\quad q\subq\starred{\flt}
    }{
      \G[\flt]\ts t : \ty[q]{T}
    }
\end{minipage}
\caption{Typing rules of \maybelang. }
\label{fig:maybe:typing}
\end{mdframed}
\end{figure}

\begin{figure}[t]\small
\begin{mdframed}
\judgement{Subtyping}{\BOX{\strut\G \ts q <: q}\ \BOX{\strut \G\ts\ty{T} <: \ty{T}}\ \BOX{\strut \G\ts\ty{Q} <: \ty{Q}}}\\
\begin{minipage}[t]{.45\linewidth}\vspace{0pt}
  \typicallabel{q-trans}
  \infrule[q-sub]{p\subq q\subq \starred{\dom(\G)}}{\G\ts p <: q}
\vgap
  \infrule[q-self]{f : \ty[q]{T}\in\G\quad\quad\QFresh\notin q}{\G\ts p,q,f <: p,f}
\vgap
  \infrule[q-var]{\ \\
    x : \ty[q]{T}\in\G\quad\quad\QFresh\notin q
  }{
    \G\ts p,x <: p,q
  }
\vgap
  \infrule[q-trans]{\G\ts p <: q\quad\quad\G\ts q <: r}{\G\ts p <: r}
\end{minipage}%
\begin{minipage}[t]{.55\linewidth}\vspace{0pt}
  \infrule[s-base]{\ \\}{
    \G\ts\ty{B} <: \ty{B}
  }
\vgap
  \infrule[s-ref]{
    \G\ts\ty{S} <: \ty{T}  \quad
    \G\ts\ty{T} <: \ty{S}\quad q\subq\DOM(\Gamma)
  }{
    \G\ts\ty{\TRef~\ty[q]{S}} <: \ty{\TRef~\ty[q]{T}}
  }
\vgap
   \infrule[s-fun]{
    \G\ts\ty{P} <: \ty{O} \\
    \G\, ,\, f : \ty[\QFresh]{(f(x : O)\to Q)}\, ,\, x : \ty{P}\ts \ty{Q} <: \ty{R}
  }{
    \G\ts\ty{f(x: \ty{O}) \to \ty{Q}} <: \ty{f(x: \ty{P}) \to \ty{R}}
  }
\vgap
  \infrule[sq-sub]{
\G\ts\ty{S} <: \ty{T}\quad\quad \G\ts p <: q
  }{
    \G\ts\ty[p]{S} <: \ty[q]{T}
  }
\end{minipage}
\caption{Subtyping rules of \maybelang. }
\label{fig:maybe:subtyping}
\end{mdframed}
\end{figure}

\begin{figure}[t]
\begin{mdframed}\small
\judgement{Qualifier Substitution and Growth}{\BOX{q[p/x]}\ \BOX{q[p/\qfresh]}}
  $$
  \begin{array}{l@{\;}c@{\;}ll@{\quad\qquad\qquad\qquad}l@{\;}c@{\;}ll}
    q[p/x] & = & q\setminus\{x\}\qlub p& x\in q    & q[p/\qfresh] &=& q\qlub p & \qfresh\in q  \\
    q[p/x] & = & q                     & x\notin q & q[p/\qfresh] &=& q & \qfresh\notin q
  \end{array}$$
\judgement{Reachability and Overlap}{\BOX{{\color{gray}\G\vdash}\,x \reaches x}\ \BOX{{\color{gray}\G\vdash}\, \qsat{q}}\ \BOX{{\color{gray}\G\vdash}\,p \overlap q}}
  $$
  \begin{array}{ll@{\qquad\qquad}ll}
   \text{Reachability Relation} & {\color{gray}\G\vdash}\, x \reaches y \Leftrightarrow  x : T^{q,y} &
   \text{Variable Saturation} & {\color{gray}\G\vdash}\, \qsat{x} := \left\{\, y \mid x \reaches^* y\, \right\} \\[1.1ex]
   \text{Qualifier Saturation} & {\color{gray}\G\vdash}\, \qsat{q} :=  \bigcup_{x\in q} \qsat{x} &
  \text{Qualifier Overlap} & {\color{gray}\G\vdash}\,p \overlap q  := \starred{(\qsat{p} \qglb \qsat{q})}
  \end{array}$$
\caption{Operators on qualifiers.  We often leave the context \(\G\) implicit (marked as gray).}
\label{fig:saturation_overlap}
\end{mdframed}
\end{figure}

\begin{figure}\small
\begin{mdframed}
\judgement{Term Typing}{\BOX{\cx[\flt]{[\Gamma\mid\Sigma]} \ts t : \ty{Q}}}
$$
  \ell \in \Loc \qquad
  \Sigma ::= \varnothing \mid \Sigma,\ell : Q \qquad
  p,q,r \subq \mathcal{P}_{\mathsf{fin}}(\Var \uplus \Loc \uplus \{ \QFresh \}) \qquad
  \flt \subq \mathcal{P}_{\mathsf{fin}}(\Var \uplus \Loc)
$$
{
\infrule[t-loc]{
  \Sigma(\ell) = \ty[q]{T}\quad q\subq\DOM(\Sigma)\quad\FV(T)=\varnothing\quad\HLBox[gray!20]{\FTV(T)=\varnothing} \quad q,\ell \subq \flt
}{
  \cx[\flt]{[\G\mid\Sigma]}\ts \ell : \ty[q,\ell]{(\TRef~\ty[q]{T})}
}}

\judgement{Location Reachability, Location \& Qualifier Saturation}{\BOX{{\color{gray}\G\mid\Sigma\vdash}\,\ell \reaches \ell}\ \BOX{{\color{gray}\G\mid\Sigma\vdash}\,\qsat{\ell}}\ \BOX{{\color{gray}\G\mid\Sigma\vdash}\, \qsat{q}}}
$$
  \textstyle {\color{gray}\G\mid\Sigma\vdash}\, \ell\reaches\ell' \Leftrightarrow \Sigma(\ell)= \ty[q,\ell']{T}
  \qquad    {\color{gray}\G\mid\Sigma\vdash}\, \qsat{\ell} := \left\{\ell'\mid \ell\reaches^* \ell' \right\}
  \qquad    {\color{gray}\G\mid\Sigma\vdash}\, \qsat{q}  :=  \bigcup_{x\in q} \qsat{x} \cup \bigcup_{\ell\in q}\qsat{\ell}
$$

\judgement{Well-Formed Stores}{\BOX{\WF{\Sigma}}}
$$
  \inferrule{\ }{\WF{\varnothing}}\qquad\qquad
  \inferrule{\WF{\Sigma}\quad \FV(T)=\varnothing\quad\HLBox[gray!20]{\FTV(T)=\varnothing} \quad \varnothing\mid\Sigma\ts \qsat{q} = q\quad \ell \notin\DOM(\Sigma) }{\WF{\Sigma\, ,\, \ell : \ty[q]{T}}}
$$

\judgement{Reduction Contexts, Values, Terms, Stores}{}
$$
  \begin{array}{l@{\ \ }c@{\ \ }l@{\qquad\qquad\ }l@{\ \ }c@{\ \ }l}
    {C} & ::= & \square \mid C\ t \mid v\ C \mid \tref~C \mid\ !{C} \mid {C} := {t} \mid {v} := {C} \mid \HLBox[gray!20]{C\ [Q]} &  t      & ::= & \cdots \mid \ell\\
    {v} & ::= & \lambda f(x).t \mid {c} \mid {\ell} \mid \tunit \mid \HLBox[gray!20]{\Lambda f(\ty[x]{X}).t}                     &  \sigma & ::= & \varnothing \mid \sigma, \ell\mapsto v
  \end{array}
$$

\judgement{Reduction Rules}{\BOX{t \mid \sigma \to t \mid\sigma}}
$$
  \begin{array}{r@{\ \ }c@{\ \ }ll@{\qquad\qquad}r}
    \CX[gray]{C}{(\lambda f(x).t)\ v} \mid \sigma     & \to & \CX[gray]{C}{t[v/x, (\lambda f(x).t)/f]} \mid \sigma                  &                           & \rulename{$\beta$} \\
    \CX[gray]{C}{\tref~v} \mid \sigma                 & \to & \CX[gray]{C}{\ell} \mid (\sigma, \ell \mapsto v)                      & \ell \not\in \DOM(\sigma) & \rulename{ref} \\
    \CX[gray]{C}{!\ell} \mid \sigma                   & \to & \CX[gray]{C}{\sigma(\ell)} \mid \sigma                                & \ell  \in \DOM(\sigma)    & \rulename{deref} \\
    \CX[gray]{C}{\ell := v} \mid \sigma               & \to & \CX[gray]{C}{\tunit} \mid \sigma[\ell \mapsto v]                      & \ell \in \DOM(\sigma)     & \rulename{assign}\\
\CX[gray]{C}{\HLBox[gray!20]{(\Lambda f(\ty[x]{X}).t)\ Q}} \mid \sigma & \to &  \CX[gray]{C}{\HLBox[gray!20]{t[Q/\ty[x]{X}, (\Lambda f(\ty[x]{X}).t)/f]}} \mid \sigma &                           & \HLBox[gray!20]{\rulename{$\beta_T$}} \\
  \end{array}
$$
\caption{Extension with store typings and call-by-value reduction for \maybelang (\Cref{sec:maybe}) and \HLBox[gray!20]{\polylang} (\Cref{sec:poly}).}\label{fig:maybe:semantics}
\end{mdframed}
\end{figure}
 
\subsection{Syntax}\label{sec:maybe:syntax}

\Cref{fig:maybe:syntax} shows the syntax of \maybelang which is based on the simply-typed
\(\lambda\)-calculus
with mutable references and subtyping.  We denote general term variables by the meta variables
\(x, y, z\),
and reserve \(f,g,h\)
specifically for function self-references in contexts where the distinction matters.

Terms consist of constants of base types, variables, recursive functions \(\lambda f(x).t\)
(binding the self-reference \(f\)
and the argument \(x\)
in the body \(t\)), function applications, reference allocations, dereferences, and assignments.

Reachability qualifiers \(p,q,r\)
are finite sets of variables that may additionally include the distinct freshness marker
\(\QFresh\).
Once we add store typings (\Cref{sec:maybe:dynamic}), qualifiers will include store locations in
addition to variables. For readability, we often drop the set notation for qualifiers
and write them down as comma-separated lists of atoms.

We distinguish ordinary types \(T\)
from qualified types \(Q = \ty[q]{T}\),
where the latter annotates a qualifier $q$ to an ordinary type $T$. The types
consist of base types \(B\)
(\eg, \Type{Int}, \TUnit), references, and dependent function types \(f(x:P) \to Q\),
where both argument and return type are qualified. The codomain \(Q\)
may depend on both the self-reference
\(f\)
and argument \(x\)
in its qualifier and type. We could alternatively separate self-references from function types
using DOT-style first-class self types~\cite{DBLP:conf/oopsla/RompfA16}, \eg, \(\mu f. ((x : P) \to Q[f,x])\).

Mutable reference types \(\TRef~Q\)
track the known aliases of the value pointed to by the reference.
We also permit forms of nested references,
which are prohibited in the base \(\lambda^*\)-calculus
unless a flow-sensitive effect system is added~\cite{DBLP:journals/pacmpl/BaoWBJHR21}.

An \emph{observation} \(\flt\)
is a finite set of variables which is part of the term typing judgment (\Cref{sec:maybe:static}).
It specifies which variables in the static environment \(\Gamma\)
are observable.  The latter assigns qualified typing assumptions to variables.

\subsection{Static Semantics}\label{sec:maybe:static}

The term typing judgment \(\G[\flt]\ts t : Q\)
in \Cref{fig:maybe:typing} states that term \(t\)
has qualified type \(Q\)
and may only access the typing assumptions of \(\Gamma\)
observable by \(\flt\). For  \(Q = \ty[q]{T}\), one may think of \(t\) as a computation that
yields a result value of type \(T\) aliasing no more than \(q\), if it terminates.
Alternatively, we could formulate the typing judgment without internalizing \(\flt\),
and instead have an explicit context filter operation
\(\G[\flt] := \{x : \ty[q]{T}\in\G \mid q,x \subq \flt \}\)
for restricting the context in subterms, just like \citet{DBLP:journals/pacmpl/BaoWBJHR21} who
loosely take inspiration from substructural type systems. Internalizing \(\flt\)
(1) makes observability an explicit notion, which facilitates reasoning about
separation and overlap, and (2) greatly simplifies the Coq mechanization.
Context filtering is only needed for term typing, but not for subtyping, so as
to keep the formalization simple.

\subsubsection{One-Step Reachability}\label{sec:laziness}
Term typing usually assigns \emph{minimal} qualifiers in the currently observable context. For
instance, term variables \(x\)
track exactly themselves \rulename{t-var}, and can be used only if they are observable
(\(x\in\flt\)).
Similarly, constants of base types are untracked \rulename{t-cst}.  We can further scale up the
qualifier to include transitively reachable variables by subsumption \rulename{t-sub} if needed.
This ``one-step'' treatment of reachability is sufficient for soundness, and shows that most of the time,
we do not have to track fully transitive reachability, but instead may compute it on-demand where it
matters, \ie, when checking separation and overlap in function applications (discussed further below).
In contrast, \citet{DBLP:journals/pacmpl/BaoWBJHR21} implicitly ensures fully transitive reachability,
\ie, term typing always assigns transitively closed qualifiers.\footnote{Cf. their
  mechanization of this variant
  \url{https://github.com/TiarkRompf/reachability/tree/main/base/lambda_star_overlap}.}
Their \rulename{t-var} rule would assign \(\ty[q,x]{T}\)
where \(q\)
is transitively closed. One-step reachability simplifies the system
and adds finer-grained precision over transitive reachability, since we can refine each step in a reachability chain
as more information is discovered during evaluation. Dependent function application and
abstraction with function self-references are prime examples (\Cref{sec:depend-appl}).

\subsubsection{Functions and Lightweight Polymorphism}\label{sec:fun-lightweight-poly}
Function typing \rulename{t-abs} implements the observable separation guarantee (cf.~\Cref{sec:motiv:funtype}),
\ie, the body \(t\) can only observe what the function type's qualifier \(q\)
specifies, plus the argument \(x\) and self-reference \(f\),
and is otherwise oblivious to anything else in the environment.
We model this by setting the observation to \(q,x,f\)
when typing the body. Thus, its observation \(q\)
at least includes the free variables of the function.
To ensure well-scopedness, \(q\) must be a subset of the observation \(\flt\).
In essence, a function type \emph{implicitly} quantifies over anything that is not
observed by \(q\), achieving a lightweight form of qualifier polymorphism.

\subsubsection{Qualifier Substitution and Growth}\label{sec:qual-subst-growth}
The base substitution operation \(q[p/x]\) of qualifiers for variables is defined in \Cref{fig:saturation_overlap},
and we use it along with its homomorphic extension to types in dependent function application.
Substitution replaces the variable with the given qualifier, if present in the target.
We suggestively overload the substitution notation for \emph{qualifier growth} \(q[p/\qfresh]\).
Capturing the intuition behind the freshness marker \(\qfresh\), growth
adds \(p\) to \(q\) only if \(\qfresh\) is present, and otherwise ignores \(p\).
Growth abstracts over reduction steps that may allocate new reachable store locations in type preservation (\Cref{thm:soundness}).
We do not remove \(\qfresh\) to permit continuous growth.

\subsubsection{Dependent Application, Separation and Overlap}\label{sec:depend-appl}
Function applications are typeable by rules \rulename{t-app} and \rulename{t-app$\QFresh$}.
The former rule applies if the function's parameter is non-fresh (\(\QFresh\notin p\))
and it matches the argument, \ie, the argument qualifier reaches only bound variables and will not
increase at run time. Applications in \rulename{t-app} are dependent, substituting the function and
argument variable in the type and qualifier of the codomain with the given qualifiers
(see \Cref{sec:motiv:depapp}).

Rule \rulename{t-app$\QFresh$}
applies to cases where the argument's qualifier is bigger than what the function type assumes, or is
expected to grow bigger due to the freshness marker \(\QFresh\).
These cases require more nuanced treatment and restrictions on the degree of dependency in the
codomain. That is, if the argument or function is fresh, then the codomain's type \(\ty{U}\)
may not be dependent on the respective variable. Otherwise, type preservation is lost
due to the potential growth with fresh runtime locations.
In total, there are four possible cases, and
we discuss two of them as specialized rules below (other cases are analogous). If neither the argument
nor the function is fresh, we obtain
$$
\inferrule{
   \Gamma \ts t_1 : \ty[q]{\left(f(x: \ty[p \overlap q]{T}) \to \ty{Q}\right)} \qquad
    \Gamma \ts t_2 : \ty[p]{T}\qquad \QFresh\notin p \qquad \QFresh\notin q
  }{
    \Gamma \ts t_1~t_2 : Q[p/x, q/f]
  }
  \quad (\textsc{t-dapp})
$$
which permits unconstrained dependency in the codomain.
If both the argument and function are fresh, we obtain
$$
\inferrule{
    \Gamma \ts t_1 : \ty[\starred{q}]{\left(f(x: \ty[p \overlap q]{T}) \to \ty[r]{U}\right)} \qquad
    \Gamma \ts t_2 : \ty[\starred{p}]{T}\qquad \{x,f\}\cap\FV(U)= \varnothing
  }{
    \Gamma \ts t_1~t_2 : \ty[{r[\starred{p}/x, \starred{q}/f]}]{U}
  }
  \quad (\textsc{t-ndapp})
$$
which requires that neither \(x\) nor \(f\) occur freely the codomain type \(U\)
(as in \citeauthor{DBLP:journals/pacmpl/BaoWBJHR21}).

In all instances of \rulename{t-app$\QFresh$},
since \(p\)
is potentially bigger than the function codomain, we need to check for \emph{observable
  separation/overlap} between function and argument, \ie, the portion of \(p\)
that the function can observe should conform with the function parameter.  This is the only place in
the type system requiring fully reflexive-transitive reachability using the \emph{overlap operator}
\(p\overlap q\)
(\Cref{fig:saturation_overlap}), which is the intersection of the smallest saturated reachability
sets of \(p\)
and \(q\), always including \(\qfresh\) to indicate that the argument is allowed to have
a bigger qualifier than the domain.
For the type safety proof, it is also sufficient to just demand
any saturated supersets.

Both function application rules impose an observability restriction on the codomain qualifier
\(r \subq \starred{\flt},x,f\),
which is to ensure that the resulting qualifier of term typings is always observable under \(\flt\)
(\Cref{lem:has_type_filter}), a critical property for the substitution lemmas and type soundness
proof.

\subsubsection{Mutable References}\label{sec:mutable-refs}
The \(\lambda^*\)
system by \citet{DBLP:journals/pacmpl/BaoWBJHR21} cannot express nested references without the
addition of a flow-sensitive effect system. Although extending it with an effect system is possible,
our type system readily supports a limited form of nested references by
means of reachability and the fresh/non-fresh distinction.  Qualifiers in reference types need to be
non-fresh in \rulename{t-ref}, \rulename{t-deref}, and \rulename{t-assgn}.
On the outside, reference allocations \rulename{t-ref} track the referent's non-fresh qualifier and \(\QFresh\),
because the final result will be a fresh new store location, which will be added to the qualifier. A
limitation of this model is the invariance of the referent's qualifier, so that only values with
identical qualifier can ever be assigned in \rulename{t-assgn}.
Therefore the referent's qualifier must be chosen large enough when introduced.
Invariance is also reflected in the
subtyping rule for references, discussed next.

\subsubsection{Subtyping}\label{sec:subtyping}
We distinguish subtyping between qualifiers \(q\), ordinary types \(T\), and qualified types \(Q\),
where the latter two are mutually dependent. Subtyping is assumed to be well-scoped under the typing
context \(\G\), \ie, types and qualifiers mention only variables bound in \(\G\),
and so do its typing assumptions.  Qualified subtyping \rulename{sq-sub} just forwards to the other
two judgments for scaling the type and qualifier, respectively.

\paragraph{Qualifier Subtyping}
Qualifier subtyping includes the subset relation \rulename{q-sub}, the two contextual rules
\rulename{q-self} and \rulename{q-var}, and transitivity \rulename{q-trans}.
Rule \rulename{q-self} is inherited from \citet{DBLP:journals/pacmpl/BaoWBJHR21}, and used for
abstracting the qualifiers of escaping closures (see examples in
\Cref{sec:motiv:funtype} and \Cref{sec2:typeabs}), \ie, if a function self
reference \(f\) and its assumed qualifier \(q\) occur in some qualifier context, then we may delete \(q\)
and just retain \(f\), because \(q\) may contain captured variables that are not visible in an outer scope.
Rule \rulename{q-var} is new here and critical for one-step reachability: a qualifier \(p,x\)
is more precise than \(p,q\) since substitution may replace \(x\) with a smaller qualifier than \(q\)
later (cf.~\Cref{sec:motiv:precise}).  This is only valid if \(\QFresh\notin q\),
because otherwise, \(x\) could be replaced later with a larger set than \(q\) and we would lose track of it.
The same restriction applies to \rulename{q-self}.

\paragraph{Ordinary Subtyping}
Subtyping rules for base types \rulename{s-base}, reference types
\rulename{s-ref}, and function types \rulename{s-fun} are standard modulo
qualifiers.
Reflexivity and transitivity are both admissible for subtyping on
ordinary and qualified types.
References are invariant in the enclosed qualifier and
equivalent in the value, expressed by bidirectional subtype constraints.
Function types are contravariant in the domain, and covariant in the codomain, as usual.  Due to dependency in
the codomain, we are careful to extend the context with the smaller argument type and self
reference.  Importantly, the function self-reference added to the context only carries the
\(\QFresh\) marker.
This distinguishes computationally relevant self-references introduced by term
typing in \rulename{t-abs} from synthetic ones for subtyping. Only the former
is eligible for abstraction by function self-references.

\subsection{Dynamic Semantics and Stores}\label{sec:maybe:dynamic}

The \maybelang-calculus adopts the standard call-by-value reduction of the \(\lambda\)-calculus
with mutable references and a store (\Cref{fig:maybe:semantics}). Term typing and subtyping change accordingly to
include store typings \(\Sigma\),
and both qualifiers and observations may now include store locations from \(\DOM(\Sigma)\).
Typing a location value \rulename{t-loc} requires that it be observable, along with the full
qualifier of the referent (\(q,\ell\subq\flt\)).
This model implements the fully transitive reachability notion for store locations instead of
one-step reachability (in contrast to variables, \Cref{sec:laziness}), as we never substitute store locations
and thus do not alter the assumed qualifiers in the store typing \(\Sigma\). The well-formedness
predicate \(\WF{\Sigma}\)  ensures that all assumptions in \(\Sigma\)
are closed and have transitively closed qualifiers consisting only of other store locations.
Well-formedness is required by \Cref{coro:preservation_separation} to ensure fully disjoint
reachability chains and object graphs.

\subsection{Metatheory}\label{sec:maybe:theory}

The \maybelang-calculus exhibits syntactic type soundness which we prove by standard progress and
preservation properties (\Cref{thm:progress,thm:soundness}). Type soundness implies the preservation
of separation corollary (\Cref{coro:preservation_separation}) as set forth by
\citet{DBLP:journals/pacmpl/BaoWBJHR21} for their \(\lambda^*\)-calculus.
It is a memory property certifying that the results of well-typed \maybelang terms with disjoint
qualifiers indeed never alias. Below, we discuss key lemmas required for the
type soundness proof, which has been proved in Coq.
Due to space limitations, we elide standard properties such as weakening
and narrowing.

\subsubsection{Observability Properties} \label{sec:maybe:obs}
Reasoning about substitutions and their interaction with overlap/separation in preservation lemmas
requires that the qualifiers assigned by term typing
are observable. The following lemmas are proved by induction over the respective typing derivations:
\begin{lemma}[Observability Invariant]\label{lem:has_type_filter}
Term typing always assigns observable qualifiers, \ie,
if $\cx[\flt]{[\G\mid\Sigma]} \ts t : \ty[q]{T}$, then $q\subq \starred{\flt}$.
\end{lemma}
\noindent Well-typed values cannot observe anything about the context beyond their assigned qualifier:
\begin{lemma}[Tight Observability for Values]\label{lem:values_tight}
If \(\ \cx[\flt]{[\G\mid\Sigma]} \ts v : \ty[q]{T}\), then \(\cx[q]{[\G\mid\Sigma]} \ts v : \ty[q]{T}\).
\end{lemma}
\noindent It is easy to see that any observation for a function \(\lambda f(x).t\)  will at least track
the free variables of the body \(t\).
Finally, well-typed values are always non-fresh in the following sense:
\begin{lemma}[Values are Non-Fresh]\label{lem:values_nonfresh}
If \(\ \cx[\flt]{[\G\mid\Sigma]} \ts v : \ty[q]{T}\), then \(\cx[\flt]{[\G\mid\Sigma]} \ts v : \ty[q\setminus\QFresh]{T}\).
\end{lemma}
\noindent This lemma is important for substitution, and asserts that values
only reach statically fully known variables and locations in context. That is, we may safely assume
that values are never the source of \(\QFresh\), and it can only stem from subsumption, which we
may undo by \Cref{lem:values_nonfresh}. Ruling out \(\qfresh\) for values ensures
that we do not accidentally add it when it is expected to be absent in a substitution target \(q\).
The absence indicates that a substitution on \(q\) will not increase it with fresh locations.

\subsubsection{Substitution Lemma}
We consider type soundness for closed terms and apply ``top-level'' substitutions, \ie, substituting
closed values with qualifiers that do not contain term variables, but only store locations.
The proof of the substitution lemma  critically relies on the distributivity of
substitution and the overlap operator (\Cref{fig:saturation_overlap}), which
is required to proceed in the \rulename{t-app\QFresh} case:
\begin{lemma}[Top-Level Substitutions Distribute with Overlap]\label{lem:subst_commutes_overlap}
$$
\infer{
  x: \ty[q]{T}\in\G\qquad \theta = [p/x] \qquad p,q\subq\starred{\dom(\Sigma)} \qquad p \qglb \starred{\flt} \subq q\qquad \qsat{r},\qsat{r'}\subq\starred{\flt}
}{
  (r \overlap r')\theta = r\theta\overlap r'\theta
}
$$
\end{lemma}
\noindent Qualifier substitution does not generally distribute with set intersection, due to the problematic
case when the substituted variable \(x\) occurs in only one of the saturated sets \(\qsat{r}\)
and \(\qsat{r'}\). Distributivity holds if (1) we ensure that what is observed about the qualifier \(p\) we substitute for \(x\)
is bounded by what the context observes about \(x\), \ie, \(p \qglb \starred{\flt}\subq q\)
for \(x : \ty[q]{T}\in\G\), and (2) \(p,q\) are top-level as above.

In the type preservation proof, $\beta$-reduction substitutes both the
function parameter and self-reference in \rulename{t-abs} (\Cref{fig:maybe:typing}) for some values.
The two substitutions can be expressed by sequentially applying a general substitution lemma on one
variable:
\begin{lemma}[Top-Level Term Substitution]\label{lem:subst_term}
$$
\infer{%
   \cx[\flt]{[\G,x:\ty[q]{T}\mid\Sigma]} \ts t : \ty{Q}\qquad \cx[p]{[\varnothing\mid\Sigma]}\ts v : \ty[p]{T}\qquad \theta = [p/x]\\
    p,q\subq\starred{\dom(\Sigma)}\qquad p \qglb \starred{\flt} \subq q \qquad q = p \vee q = \starred{(p \qglb r)}
}{
    \cx[\flt\theta]{[\Gamma\theta\mid\Sigma]} \ts t[v/x] : \ty{Q}\theta
}
$$
\end{lemma}
\begin{proof}
By induction over the derivation \(\cx[\flt]{[\G,x:\ty[q]{T}\mid\Sigma]} \ts t : \ty{Q}\).
Most cases are straightforward, exploiting that qualifier substitution is monotonic w.r.t.\
\(\subq\) and that the substitute \(p\) for \(x\) consists of store locations only.
The case \rulename{t-app\QFresh} critically requires \Cref{lem:subst_commutes_overlap} for
\((p \overlap q)\theta = p\theta \overlap q\theta\) in the induction hypothesis.
The case \rulename{t-sub} requires an analogous substitution lemma for subtyping (elided due to space
limitations).
\end{proof}
\noindent Just as above, the substitution lemma imposes the observability condition \(p \qglb \starred{\flt} \subq q\).
The condition \(q = p \vee q = \starred{(r \qglb p)}\) captures the two different cases
of substitution: (1) a precise substitution where the assumed qualifier \(q\) for \(x\) is identical to the
value's qualifier \(p\), \ie, the parameter in \rulename{t-app} or the function's self-reference
\(f\) in \rulename{t-app}/\rulename{t-app\qfresh},
or (2) a growing substitution for the parameter in \rulename{t-app\qfresh} with
overlap between \(p\) and the function qualifier \(r\), growing the result
by  \(p\setminus \qsat{r}\).

\subsubsection{Main Soundness Result}\label{sec:maybesoundness}

\begin{theorem}[Progress]\label{thm:progress}
 If \(\ \cx[\DOM(\Sigma)]{[\varnothing\mid\Sigma]}  \ts t : \ty{Q}\), then either \(t\) is a value, or
 for any store \(\sigma\) where \(\varnothing \mid \Sigma \ts \sigma\), there exists
 a term \(t'\) and store \(\sigma'\) such that \(t \mid \sigma \to t' \mid \sigma'\).
\end{theorem}
\begin{proof}
By induction over the derivation \(\cx[\DOM(\Sigma)]{[\varnothing\mid\Sigma]}  \ts t : \ty{Q}\).
\end{proof}
\noindent Similar to \cite{DBLP:journals/pacmpl/BaoWBJHR21}, reduction preserves types up to qualifier growth
(cf.~\Cref{sec:qual-subst-growth}):
\begin{theorem}[Preservation]\label{thm:soundness}
  If \(\ \cx[\DOM(\Sigma)]{[\varnothing\mid\Sigma]}  \ts t : \ty[q]{T}\),
  and \(\varnothing \mid \Sigma \ts \sigma\), and \(t \mid \sigma \to t' \mid \sigma'\),
  then there exists \(\Sigma' \supseteq \Sigma\) and \(p \subq\DOM(\Sigma'\setminus\Sigma)\)
  such that \(\varnothing \mid \Sigma' \ts \sigma'\)
  and \(\cx[\DOM(\Sigma')]{[\varnothing\mid\Sigma']} \ts t' : \ty[{q[p/\qfresh]}]{T}\).
\end{theorem}
\begin{proof}
  By induction over the derivation \(\cx[\DOM(\Sigma)]{[\varnothing\mid\Sigma]}  \ts t : \ty[q]{T}\).
\end{proof}

\begin{corollary}[Preservation of Separation]\label{coro:preservation_separation}\
Interleaved executions preserve types and disjointness:\vspace{-10pt}

\infrule{
  \begin{array}{l@{\qquad}l@{\qquad}ll}
   \cx[\DOM(\Sigma)]{[\varnothing \mid \Sigma]} \ts t_1 : \ty[q_1]{T_1} &  t_1 \mid \sigma\phantom{'} \to t_1' \mid \sigma' & \varnothing \mid \Sigma \ts \sigma & \WF{\Sigma}\\[1ex]
\cx[\DOM(\Sigma)]{[\varnothing \mid \Sigma]} \ts t_2:\ty[q_2]{T_2} & t_2 \mid \sigma' \to t_2' \mid \sigma'' & q_1 \overlap q_2 \subq \{\vardiamondsuit\} &
  \end{array}
}{
  \begin{array}{ll@{\qquad}l@{\qquad}l}
\exists p_1\;p_2\;\Sigma'\;\Sigma''. & \cx[\DOM(\Sigma')\phantom{'}]{[\varnothing \mid \Sigma'\phantom{'}]} \ts t_1' : \ty[p_1]{T_1} & \Sigma'' \supseteq \Sigma' \supseteq \Sigma \\[1ex]
& \cx[\DOM(\Sigma'')]{[\varnothing \mid \Sigma'']} \ts t_2' : \ty[p_2]{T_2} & p_1 \overlap p_2 \subq \{\vardiamondsuit\}
  \end{array}
}
\end{corollary}
\begin{proof}
  By sequential application of Preservation (\Cref{thm:soundness}) and the fact that a reduction step
  increases the assigned qualifier by at most a fresh new location, thus preserving disjointness.
\end{proof}

\section{Reachability and Type Polymorphism}\label{sec:poly}

We extend the simply-typed reachability-polymorphic system
$\maybelang$ with type-and-qualifier abstraction in the style of $F_{<:}$
\cite{DBLP:journals/iandc/CardelliMMS94}.
The typing of this extension behaves the same as in standard \Fsub modulo
self-references and reachability sets.
As mentioned in \Cref{sec2:typeabs}, we simultaneously abstract over types and
qualifiers, because just abstracting over types leads to imprecise reachability
tracking for data-type eliminations.

\subsection{Syntax} \label{sec:qpoly:syntax}

\Cref{fig:poly} shows the syntax of \polylang as a \Fsub-style
extension of $\maybelang$.
Types now include the $\TTop$ type, type variables $X$, and universal types.
A universal type introduces a quantified type variable $X$ along with a
quantified qualifier variable $x$, which are both upper-bounded by a qualified
type $Q$.
It is important to read the combined quantification as an abbreviation
introducing the abstract type and qualifier independently, and they do
not need to be used together, \ie,
\( \forall (X <: T). \forall (x <: q). Q \equiv \forall (\ty[x]{X} <: \ty[q]{T}). Q \).
We choose the more compact syntax for readability since types and qualifiers
are often instantiated together.
Similar to function types, universal types have self-references, which are
useful when a polymorphic closure escapes its defining scope.
The body of a universal type is also qualified and can access the
self-reference $f$ of the universal type in addition to $x$.
Terms now include type abstractions and qualified type applications.
Type abstractions bind their own self-reference \(f\), type parameter \(X\), and qualifier
parameter \(x\) in the body \(t\).
Typing environments now include bounded type-and-qualifier variables
of the form \(\ty[x]{X} <: Q\).

\subsection{Static Semantics} \label{sec:qpoly:static}

The typing and subtyping rules of \polylang (\Cref{fig:poly}) are a
superset of those presented for \maybelang in \Cref{sec:maybe}.

\subsubsection{Typing Rules}

We add the typing rules for type abstractions and type applications.
The type system is defined declaratively in Curry-style, and hence
for type abstractions \rulename{t-tabs} we need to ``guess'' the whole universal
type and its qualifier.
Other parts are analogous to term abstraction typing (\Cref{sec:fun-lightweight-poly}).
Notably, observable separation naturally generalizes to type abstraction. That is,
the qualifier \(q\) constrains what the type abstraction's implementation can observe,
and \(P\)'s qualifier in \(\ty[x]{X}<: P\) determines observable overlap/separation
for instantiations of \(x\). Especially, if \(P\) mentions the freshness marker \qfresh,
then instantiations of \(x\) can mention unobserved variables.

Similar to function applications in \maybelang, there are two type application
rules: \rulename{t-tapp} for non-fresh dependent applications and
\rulename{t-tapp$\QFresh$} for restricted dependent applications.
Requiring non-freshness ensures that we pass a qualifier argument that
is bounded by other variables in the context.
Rule \rulename{t-tapp$\QFresh$}
is analogous to \rulename{t-app$\QFresh$}
(cf.~\Cref{sec:depend-appl}): If the argument/function qualifier is fresh, then the result
type \(U\) cannot be dependent on it. We impose observability constraints
on the codomain qualifier \(r\) to ensure the
observability invariant (\Cref{lem:has_type_filter})  for \polylang{}.

\begin{figure}[t]\small
\begin{mdframed}
\begin{minipage}[t]{1.0\textwidth}\small
  \judgement{Syntax}{\BOX{\polylang}}\vspace{-8pt}
  \[\begin{array}{l@{\qquad}l@{\qquad}l@{\qquad}l}
    T &::=& \dots \mid \TTop \mid X \mid \forall f(\ty[x]{X} <: Q). Q & \text{Types} \\
    t &::=& \dots \mid \TLam{X}{x}{T}{q}{t} \mid \TApp{t}{Q}{} & \text{Terms} \\
    \Gamma &::=& \dots \mid \Gamma, \ty[x]{X} <: Q & \text{Typing Environments}
  \end{array}\]
  \judgement{Term Typing}{\BOX{\strut\G[\flt] \ts t : \ty{Q}}}\\[1ex]
  \begin{minipage}[t]{1.0\linewidth}\small
    \typicallabel{t-tapp$\QFresh$}
    \infrule[t-tabs]{
      \cx[q,x,f]{\left(\G\ ,\ f: \ty{F}\ ,\ \ty[x]{X} <: \ty{P}\right)} \ts t : \ty{Q} \qquad
      F = \ty[q]{\left(\TAll{X}{x}{P}{}{Q}{}\right)}
      \qquad q \subq \varphi
    }{
      \G[\flt] \ts \TLam{X}{x}{T_1}{q_1}{t} : F
    }
  \vgap
    \infrule[t-tapp]{
      \G[\flt] \ts t : \ty[q]{\left(\TAll{X}{x}{T}{p}{Q}{}\right)} \qquad
      \QFresh \notin p \qquad \\
      p \subseteq \varphi \qquad
      r\subq\starred{\flt},x,f \qquad
      Q = \ty[r]{U} \qquad
    }{
      \G[\flt] \ts t [ \ty[p]{T} ] : Q[\ty[p]T/\ty[x]{X}, q/f]
    }
  \vgap
    \infrule[t-tapp$\QFresh$]{
      \G[\flt] \ts t : \ty[q]{\left(\TAll{X}{x}{T}{p \overlap q}{Q}{}\right)} \qquad
      \QFresh \in p \Rightarrow x\notin\FV(U) \qquad
      \QFresh \in q \Rightarrow f\notin\FV(U) \\
      p \subseteq \varphi \qquad
      r\subq\starred{\flt},x,f \qquad
      Q = \ty[r]{U} \qquad
    }{
      \G[\flt] \ts t [ \ty[p]{T} ] : Q[\ty[p]{T}/\ty[x]{X}, q/f]
    }
\end{minipage}\\[1ex]
  \judgement{Subtyping}{\BOX{\G \ts q <: q}\ \BOX{\G\ts\ty{T} <: \ty{T}}}\\[1ex]
\begin{minipage}[t]{.4\linewidth}
    \typicallabel{q-qvar}
    \infrule[q-qvar]{
      \ty[x]{X} <: \ty[q]{T} \in \Gamma \qquad
      \QFresh \notin q
    }{
      \Gamma \ts p,x <: p,q
    }
    \vgap
     \infrule[s-tvar]{\\
      \ty[x]{X} <: \ty[q]{T} \in \Gamma
    }{
      \Gamma \ts X <: T
    }
\end{minipage}
\begin{minipage}[t]{.03\linewidth}
\hspace{1pt}
\end{minipage}%
\begin{minipage}[t]{.57\linewidth}\small
  \typicallabel{s-top}
      \infrule[s-top]{}{
      \Gamma \ts T <: \TTop
    }
    \vgap
    \infrule[s-all]{
       \Gamma \ts Q <: O\\
      \Gamma\ ,\ f: \ty[\qfresh]{\left(\TAll{X}{x}{O}{}{P}{}\right)}\ ,\ \ty[x]{X} <: Q \ts P <: R
    }{
      \Gamma \ts \TAll{X}{x}{O}{}{P}{} <: \TAll{X}{x}{Q}{}{R}{}
    }
\end{minipage}
\end{minipage}
\caption{The syntax and typing rules of \polylang as an extension of \maybelang.}%
\label{fig:poly}
\end{mdframed}
\end{figure}

\subsubsection{Subtyping Rules}

Rule \rulename{s-top} is the standard rule for the $\TTop$ type.
Instead of defining a single subtyping rule for type-and-qualifier variables
that simply looks up the context in the premise, we disentangle it
to the \rulename{q-qvar} rule and \rulename{s-tvar} rule, distinguishing
the subtyping for qualifiers and ordinary types (cf.~\Cref{sec:subtyping}).
The former accounts for subtyping of qualifiers, allowing upcasting a
qualifier variable to its upper bound.
The latter is akin to standard type variable subtyping.
This disentanglement reflects the fact that we can upcast the quantified
qualifier and type independently, despite that they are introduced
together using a combined syntax.

For universal types \rulename{s-all}, we use the ``full'' subtyping rule for
richer expressiveness \cite{DBLP:journals/mscs/CurienG92} where type bounds are
contravariant.
This rule renders subtyping an undecidable relation \cite{DBLP:conf/popl/Pierce92}.
However, the choice of using the ``full'' variant is not essential to our
calculus and our main metatheoretic result is type soundness and preservation
of separation.
It should be possible to obtain a decidable fragment by building atop the
``kernel'' variant of \Fsub.  In this case, we would need to check subtyping of
the type argument explicitly, rather than relying on subsumption to upcast the
type of the type abstraction itself.
Due to self-references, we also extend the context with the smaller universal
type when subtyping the body, as in DOT~\cite{DBLP:conf/oopsla/RompfA16}.  Note
that \rulename{s-all} invokes subtyping on qualified types in its premises.

\subsection{Dynamic Semantics and Metatheory} \label{sec:qpoly:dynamics} \label{sec:qpoly:theory}

\Cref{fig:maybe:semantics} highlights the changes and new rules of
$\polylang$'s dynamic semantics, as an extension of $\maybelang$.
The reduction semantics is entirely standard compared to \Fsub, the only difference being
that type abstractions are recursive, so that type application \rulename{$\beta_T$} also substitutes
the type abstraction itself along with the argument.
Since location typing \rulename{t-loc} and store well-formedness
require closed types in store typings,  we additionally demand the absence of free type variables.

$\polylang$ enjoys the same soundness properties as $\maybelang$, \ie, progress,  preservation, and the separation of
preservation corollary (cf.~\Cref{sec:maybesoundness}). As for $\maybelang$, we have proved these results in Coq for \polylang.

\section{Comparison with Scala Capture Types}\label{sec:scalacapture}

A closely related work to ours is the recent Scala Capture Types (CT)
proposal~\cite{DBLP:conf/scala/OderskyBBLL21,DBLP:journals/corr/abs-2105-11896,
DBLP:journals/corr/abs-2207-03402} which also tracks sets of variables.
The system is tailored to programming with effects as non-escaping
capabilities, providing a lightweight form of effect polymorphism.
In this section, we inspect a few aspects of CT and compare
with reachability types.

\subsection{Capture Sets and The Universal Qualifier}

Similar to our system, capture types are built on top of \Fsub and
types can be annotated with variable sets, \ie,
$\{c_1, \dots, c_n\}\ T$ where $c_i$ is a variable representing the captured
capability.
Here is a (simplified) combinator for scoped exception handling using
capture types~\cite{DBLP:conf/scala/OderskyBBLL21}:

\begin{lstlisting}[language=Neutral]
  // declares the throw capability:
  class CanThrow
  // passes a tracked non-escaping capability to a block:
  def _try[A](block: (c: {*} CanThrow) -> A) = block(CanThrow())
\end{lstlisting}
Importantly, \lstinline[language=Neutral]|{*}| is a special marker for the top
element for qualifier subtyping in capture types, meaning some unknown set of variables is
tracked, \eg, \lstinline[language=Neutral]|{*} CanThrow| above.

While superficially similar, this top qualifier should not to be confused with
our \(\qfresh\) marker indicating a fresh/growing qualifier, and behaves
differently, as we will show later.

\subsection{Box Types for Non-Escaping Capabilities} \label{sec:ct:box}
Since @c@ represents a universal capability, we want to enforce that the
lifetime of capability \lstinline[language=Neutral]|c| passed to the given
\lstinline[language=Neutral]|block| is bound to the scope of
\lstinline[language=Neutral]|_try|.
In other words, it should not be leaked for any given
\lstinline[language=Neutral]|block|, \eg, by directly returning it or returning it indirectly
through an escaping closure. Capture types enforce this by requiring that the universal
capability @{*}@ cannot escape. This is in contrast to capabilities bound to a variable in an outer scope.

When combined with parametric type polymorphism,
\citet{DBLP:conf/scala/OderskyBBLL21, DBLP:journals/corr/abs-2207-03402}
propose to use a box type operator $\square\;T$ to turn qualified types
into proper, unqualified types, so that type variables only need to
range over proper types.
A boxed value \(\square[q\;T]\) capturing local variables in
\(q\) is upcast to \(\square[\{*\}\;T]\) when going out of scope.
Unboxing such types recovers the capture set.
Boxed values can only be unboxed if the contained qualifier \(q\) is a concrete
variable set, specifically excluding the top qualifier \lstinline[language=Neutral]|{*}|.
This provides a mechanism for statically enforcing non-escaping capabilities, \ie,
boxes are implicitly inserted at the abstraction boundary whenever the block's return
type \lstinline[language=Neutral]|A| is instantiated with a tracked type:
\begin{lstlisting}[language=Neutral]
  // illegal use (escaping capabilities):
  val x = _try { c => c }       //$\,$: $\square$[{*} CanThrow], error: cannot box/unbox
  val y = _try { c => () => c } //$\,$: $\square$[{*} (() -> CanThrow)], error : cannot box/unbox
\end{lstlisting}
On the outside, subtyping can only assign the \lstinline[language=Neutral]|{*}| qualifier to blocks
that return or capture the capability \lstinline[language=Neutral]|c|, since the captured variable
is not visible in the outer scope.

\subsection{Limitation: Tracking Fresh Values}
Let us now consider combining @_try@ with other resources
that have non-scoped introduction forms and should be tracked:
\begin{lstlisting}[language=Neutral]
  // assume freshAlloc() : {*} T
  val outer = freshAlloc()          // : {*} T is bound to {outer} T
  val z = _try { c => () => outer } // : $\square$[{outer} (() -> T)], ok: can box/unbox
\end{lstlisting}
The compiler rejects unboxing the
\lstinline[language=Neutral]|{*}| qualifier, but allows it for any more concrete one.
However, while the box type prevents capabilities from escaping, the compiler must infer and insert
box introductions and eliminations at declaration and use sites of polymorphic terms.  But more
importantly, the capture types mechanism does not support unbound fresh values well, \eg,
fresh allocations. The obvious choice is assigning the top-qualifier
\lstinline[language=Neutral]|{*}| to indicate some new value, but this is at odds with
boxing/unboxing, \eg, one cannot write
\begin{lstlisting}[language=Neutral]
  val fresh  = _try { c => freshAlloc() }                  // : $\square$[{*} T], error
  val fresh2 = _try { c => val f = freshAlloc(); () => f } // : $\square$[{*} (() -> T)], error
\end{lstlisting}
A potential workaround is having a separate global capability (\eg, @heap@) for allocations:
\begin{lstlisting}[language=Neutral]
  // fresh3 : {fresh2} T <: {heap} T
  val fresh3 = _try { c => heap.freshAlloc() } // : $\square$[{heap} T], ok: can unbox
\end{lstlisting}
This solution works well for effects-as-capabilities models, but it is
unsatisfactory for tracking aliasing and separation, \eg, all fresh values have
a common super type \lstinline[language=Neutral]|{heap} T|
which pollutes subtyping chains and leads to a loss of distinction between
separate fresh allocations.  In summary, if we want to track the lifetimes of a given class
of resources, these lifetimes must be properly nested in a stack-like manner
with the lifetimes of all other resources.

\subsection{The Reachability Approach}
Our \polylang{}-calculus can correctly handle fresh values, while at the same
time not requiring a box type.  This stems from (1) having an
intersection operator for reasoning about separation/overlap, and (2) a strong observability guarantee on
function types and universal types.
For instance, here is the type- and qualifier-polymorphic version of @_try@:
\begin{lstlisting}
  // $\forall \text{\texttt{A}}^{z} <: \TTop^\qfresh.\ ((\text{\texttt{CanThrow}}^\qfresh \to \text{\texttt{A}}^{\{z,\vardiamondsuit\}} )^\qfresh\to \text{\texttt{A}}^{\{z,\vardiamondsuit\}})$
  def _try[A$\trackfresh$](block: ((c: CanThrow$\trackfresh$) => A)$\trackfresh$): A = block(CanThrow())
\end{lstlisting}
The annotation on the @block@ parameter specifies that it is contextually fresh
for the implementation of @_try@ and thus entirely separate in terms of
transitive reachability.
We still reject the @x@ and @y@ examples above (\Cref{sec:ct:box}), but we now
permit @fresh: T$\trackfresh$@ and @fresh2: f(() => T$\trackset{f}$)@,
which correctly preserves the freshness of unnamed results.
Finally, \polylang permits finer-grained type distinctions when returning fresh
values, due to function self references:

\vspace{-10pt}
\begin{minipage}[t]{.4\linewidth}
\begin{lstlisting}
$\text{\texttt{\color{OliveGreen}// CT:\ \ \!\{*\} (() -> T)}}$
// $\polylang$: () => T$\trackfresh$
def retFresh() = () => freshAlloc()
\end{lstlisting}
\end{minipage}%
\begin{minipage}[t]{.6\linewidth}
\begin{lstlisting}
  $\text{\texttt{\color{OliveGreen}// CT:\ \ \!\{*\} (() -> T)}}$
  // $\polylang$: f(() => T$\trackset{f}$)
  def retConst() = { val f = freshAlloc(); () => f }
\end{lstlisting}
\end{minipage}\\
\polylang distinguishes between returning a fresh value on each invocation,
versus returning one and the same fresh value escaping a local scope,
whereas both are indistinguishable in capture types.

\subsubsection*{Outlook}
Compared to CT, reachability types exhibit similar expressiveness
and can support all uses of capture types.
Additionally, reachability types show richer expressiveness in a few key
aspects, especially the tracking of freshness and the guarantee of separation.
In future work, we propose to implement reachability types on top of
the experimental implementation \cite{ScalaCC}
of capture types in Scala 3, which would additionally provide a notion of separation.
We look forward to seeing how these two lines of similar ideas can benefit each other
in the future.

On the other hand, it is possible to further increase the expressiveness power
of reachability types by layering a flow-sensitive effect system on top of it.
With the notion of reachability, we can soundly express uniqueness, use-once
capabilities, and move semantics as described in \cite{DBLP:journals/pacmpl/BaoWBJHR21}.
These are useful to model low-level memory deallocation, one-shot
continuations, lock/unlock in concurrency programming, etc.

\section{Related Work} \label{sec:related}

\paragraph{Tracking Variables in Types}
The most directly related work of this paper is the original work on
reachability types \cite{DBLP:journals/pacmpl/BaoWBJHR21}.
This paper addresses the limitation of \citet{DBLP:journals/pacmpl/BaoWBJHR21}
and improves its expressiveness by introducing a new reachability tracking
mechanism, the freshness notion, and type-and-qualifier quantification.

Capture types \cite{DBLP:journals/corr/abs-2105-11896,
DBLP:journals/corr/abs-2207-03402, DBLP:conf/scala/OderskyBBLL21} is another
recent ongoing effort to integrate capability tracking and escaping checking
into Scala 3.
Several calculi have been proposed for capture types, \eg, $\mathsf{CF}_{<:}$
\cite{DBLP:journals/corr/abs-2105-11896} and $\mathsf{CC}_{<:\square}$
\cite{DBLP:journals/corr/abs-2207-03402, DBLP:conf/scala/OderskyBBLL21}.
In \Cref{sec:scalacapture}, we have discussed and compared with capture types.
To achieve capture tunnelling with universal polymorphism, the
$\mathsf{CC}_{<:\square}$ calculus uses boxing/unboxing, inspired
by contextual modal type theory (CMTT) \cite{DBLP:journals/tocl/NanevskiPP08}.
\citet{DBLP:conf/lpar/SchererH13} propose open closure types where function
types are attached with its defining lexical environment. It is used
for data flow analysis.
Several type systems \cite{DBLP:journals/pacmpl/JangGMP22,
DBLP:conf/aplas/KiselyovKS16, DBLP:journals/pacmpl/ParreauxVSK18} designed for
manipulating open code in metaprogramming also track free variables and
contexts in types, which are closely related to CMTT.

\paragraph{Escaping, Freshness, and Existential Types}
Works inspired by regions \cite{tofte1997region} use existential types for
tracking freshness or escaping entities, \eg, in Alias
types~\cite{DBLP:conf/esop/SmithWM00},
$L^3$~\cite{DBLP:journals/fuin/AhmedFM07}, and
Cyclone~\cite{grossman2002region}, analogous to our freshness marker
and self-reference.
As an analogy, one can think of a type with the freshness marker
@Ref@$\trackfresh$ as having an underlying quasi-existential type $\mu x.$@Ref@$\trackset{x}$
where the reference type tracks its own self-reference.
However, existentials for this purpose in our system would have to preserve
precise reachability information across temporary aliases created during
pack/unpack operations. That is, special facilities simulating freshness marker
and related constructs would need to be used in the implementation of
existentials, if those were taken as primitives.
Therefore, we believe the typing with self-references is more concise and
appropriate than existentials here, because we can use the \emph{same}
variable.
In addition, the use of self-references for escaping closures in our work makes
the reasoning succinct.
Similar to our calculi, type systems distinguishing second-class values can
also enforce non-escaping properties of effects or capabilities
\cite{osvald2016gentrification, DBLP:journals/pacmpl/BrachthauserSLB22,
DBLP:journals/pacmpl/BrachthauserSO20, DBLP:conf/ecoop/XhebrajB0R22, DBLP:journals/corr/abs-1201-0023}.
To regain the ability to return second-class capabilities,
\citet{DBLP:journals/pacmpl/BrachthauserSLB22} again make use of boxing and unboxing.

\paragraph{Separation}
The notion of separation and intersection operator (\Cref{sec:depend-appl})
used in reachability types is inspired by separation logic
\cite{DBLP:conf/lics/Reynolds02, DBLP:conf/csl/OHearnRY01} and its predecessors
\cite{DBLP:conf/popl/Reynolds78,reynolds1989syntactic,o1999syntactic}.
Bunched typing~\cite{DBLP:journals/jfp/OHearn03} and syntactic control of
interference \cite{DBLP:conf/popl/Reynolds78,reynolds1989syntactic,o1999syntactic}
allow reasoning about disjoint and shared resource access.
This is similar to reachability types, however, our system does not enforce that
the \emph{computations} of the function and arguments are disjoint,
but their final \emph{values} are disjoint (rule
\textsc{t-app$\QFresh$} in \Cref{fig:maybe:typing}).
Bunched typing enforces separation by splitting the typing context, whereas our work
enforces separation by checking disjointness of \emph{saturated} reachability
sets.
Bunched typing also lacks an explicit treatment of aliasing.

Uniqueness types~\cite{DBLP:journals/mscs/BarendsenS96,DBLP:conf/ifl/VriesPA06,DBLP:conf/ifl/VriesPA07}
ensure that there is no more than one reference pointing to the resource,
effectively establishing separation.
\citet{DBLP:conf/esop/MarshallVO22} present a language unifying linearity
\cite{wadler1990linear} and uniqueness.
Our base system does not directly track either linearity or uniqueness, instead,
flow-sensitive ``kill'' effects that disable all aliases can be integrated to
statically enforce uniqueness \cite{DBLP:journals/pacmpl/BaoWBJHR21}.

\paragraph{Polymorphism}
Reachability types and our variants feature lightweight reachability
polymorphism without introducing explicit quantification (cf. \Cref{sec:fun-lightweight-poly}).
Capture types \cite{DBLP:journals/corr/abs-2105-11896,
DBLP:journals/corr/abs-2207-03402, DBLP:conf/scala/OderskyBBLL21}
provide a similar flavor via dependent function application.
\citet{DBLP:journals/pacmpl/BrachthauserSO20,
DBLP:journals/pacmpl/BrachthauserSLB22} propose to represent effects as
capabilities, which yields a lightweight form of effect polymorphism
that requires little annotations.

Various forms of polymorphism exist in prior work on ownership types.
\citet{noble1998flexible} uses generic parameters to pass aliasing modes into a
class. But they do not allow ownership parameterization isolated from type
parameterization.
\citet{clarke2003object} further supports ownership polymorphism
via context parameters.
Similarly, Ownership Generic Java~\cite{DBLP:conf/oopsla/PotaninNCB06} allows
programmers to specify ownership information through type parameters.
Jo$\exists$~\cite{DBLP:conf/esop/CameronD09,DBLP:phd/ethos/Cameron09} combines
the theory of existential types with a parametric ownership type system, where
ownership information is passed as additional type arguments.
Generic Universe Types~\cite{DBLP:journals/toplas/DietlDM11} integrate the
owners-as-modifiers discipline with type genericity, effectively
separating the ownership topology from the encapsulation constraints.

\citet{DBLP:journals/mscs/CollinsonPR08} combine \textsf{F}-style polymorphism
with bunched logic, where universal types are discerned to be either additive
and multiplicative, but do not allow abstraction over additivity and
multiplicativity.
Our system $\polylang$ has quantified abstraction over qualifiers, which can be
used as an argument's reachability, permitting flexible instantiation of
either disjointness of sharing.

Constraints in alias types \cite{DBLP:conf/esop/SmithWM00} support a form of
location and store polymorphism, where the latter abstracts over irrelevant
store locations.
Our calculi implicitly abstract over contexts by baking the observability
notion into typing.

\paragraph{Ownership Types}

Ownership type systems \cite{noble1998flexible, DBLP:conf/oopsla/ClarkePN98}
are generally concerned with objects in OO programs and start from the
uniqueness restriction \cite{muller2000type, DBLP:journals/toplas/DietlDM11,
clarke2001simple, boyapati2002ownership, zhao2008implicit} and then selectively
re-introduce sharing in a controlled manner \cite{hogg1991islands,
naden2012type, clebsch2015ownership}.
Inherited from \citet{DBLP:journals/pacmpl/BaoWBJHR21}, our calculi
are designed for higher-order languages and deem sharing
and separation as essential substrates, on top of which an additional effect system
can be layered to achieve uniqueness and ownership transfer.
The focus of this paper is to address the limitations in expressiveness of
\citet{DBLP:journals/pacmpl/BaoWBJHR21} regarding reachability and type
polymorphism.

Rust's type system \cite{DBLP:conf/sigada/MatsakisK14} enforces strict
uniqueness of mutable references, while immutable references can be shared via
borrowing, known as the ``shared XOR mutable'' rule.
Mezzo \cite{DBLP:journals/toplas/BalabonskiPP16} is a language designed for
control aliasing and mutation and share some similarities with \polylang.
Mezzo tracks aliasing using singleton types \cite{DBLP:conf/esop/SmithWM00}.
When dealing with effects, Mezzo imposes restrictions like
Rust: mutable portions of the heap must have a unique owner, whereas
reachability types relax this constraint.
Moreover, Mezzo lacks the notion of separation between functions and arguments
and uses existential quantification to handle escaping functions that
capture local variables.
\polylang checks separation at the call site and has a lightweight mechanism to
track escaping functions via self-references.

Typestate-oriented programming \cite{DBLP:conf/oopsla/AldrichSSS09} and its combination with gradual typing
\cite{DBLP:journals/toplas/GarciaTWA14} also provides static flow-sensitive
reasoning or dynamic enforcement.

\section{Conclusion} \label{sec:conclusion}

In this work, we investigate limitations in expressiveness found in prior
reachability type systems \cite{DBLP:journals/pacmpl/BaoWBJHR21}.
We propose a new reachability type system $\maybelang$ that has lightweight,
precise, and sound reachability polymorphism.
Based on $\maybelang$, we add bounded quantification over types and qualifiers,
leading to a type-and-reachability-polymorphic calculus $\polylang$.
We have formalized these systems and proved the soundness and separation
guarantees in Coq.
We further discuss applying $\polylang$ to programming with capabilities and
compare with Scala capture types.
Our system subsumes both prior reachability types and the essence of
Scala capture types, while
exhibiting richer expressiveness in key aspects such as modeling freshness and
guaranteeing separation.

\bibliography{references}

\end{document}